\documentclass[twocolumn]{IEEEtran} 
\usepackage{amsfonts,color,morefloats}
\usepackage{amssymb,amsmath,latexsym,amsthm}

\usepackage{ulem}
\usepackage{comment}

\newcommand{\C}{{\mathcal{C}}}

\newtheorem{theorem}{Theorem}
\newtheorem{lemma}{Lemma}
\newtheorem{proposition}{Proposition}
\newtheorem{corollary}{Corollary}

\newtheorem{problem}{Open Problem}
\newtheorem{definition}{Definition}
\newtheorem{example}{Example}
\newtheorem{remark}{Remark}

\def\LI{\langle}
\def\RI{\rangle}
\def\F{{\mathbb F}}

\def\Aut{\mbox{Aut}}

\def\0{{\bf 0}}

\begin{document}

\title{Designs, linear codes, plateaued functions, and their interconnections \thanks{
 This work was supported by the National Natural Science Foundation of China (Grant Nos. 62372247 and 12441103) and by open research fund of State Key Laboratory of Cyberspace Security Defense (No. 2025-MS-03). (Corresponding author: Yansheng Wu)
}
}

\author{Jong Yoon Hyun, Jieun Kwon, Jiaxin Wang, Yansheng Wu\thanks{ J. Y. Hyun is with the Konkuk University, Glocal Campus, 268 Chungwon-daero Chungju-si Chungcheongbuk-do 27478, South Korea, Email: hyun33@kku.ac.kr.}
\and \thanks{  J. Kwon is with the Department of Mathematics, Pohang University of Science and Technology (POSTECH), 77 Cheongam-Ro, Nam-Gu, Pohang, Gyeongbuk 37673, South Korea, Email: je1997@postech.ac.kr.}

\and \thanks{J. Wang is with the School of Mathematics, Hefei University of Technology, Hefei, 230601, China,
 email: wjiaxin@hfut.edu.cn. }

\and \thanks{Y. Wu is  with the School of Computer Science, Nanjing University of Posts and Telecommunications, Nanjing
210023, China and the State Key Laboratory of Cyberspace Security Defense (Institute of Information Engineering, Chinese Academy of Sciences, Beijing 100085), China.  
 Email: yanshengwu@njupt.edu.cn. }

}

\date{\today}
\maketitle

{


\begin{abstract}  In this paper, we mainly investigate profound interconnections  between combinatorial designs, linear  codes, and Boolean functions. Firstly, we present a generic  construction method for designs derived from Boolean functions and give a new concept of non-symmetric designs with the triple symmetric difference property (TSDP).  

Secondly, we provides an alternative proof for addition designs derived from plateaued functions, which need not be simple or symmetric. We characterize simple $2$-designs on $2^{m-r}$ points arising from $r$-plateaued functions in $m$ variables, and show that addition designs from such functions with no nonzero linear structure satisfy the TSDP but not the double one, yielding non-symmetric simple $2$-designs.
 
 Thirdly, we primarily explore the equivalence relationships between designs, linear codes, and plateaued functions. These investigations help resolve two open problems posed by Ding and Tang  (Designs from Linear codes, Singapore: World Scientific, 2022: Problems 14.20, 14.23). We also compute the automorphism groups of addition designs of $r$-plateaued functions and  of the linear codes of addition designs. This work extends  results by Bending (SDP designs and their automorphism groups, Ph.D. thesis, 1993), and Dempwolff and Neumann (Des. Codes Cryptogr.,  57,
373–381, 2010). 

Finally, we yield new Boolean functions producing two families of a 2-design whose parameters coincide with those of the complement of a point-hyperplane design and a TSDP design, despite being non-isomorphic.

\end{abstract}

 }


\section{Introduction}\label{sec-intro} 

Bent functions, first introduced by Rothaus \cite{R}, are Boolean functions exhibiting maximum possible non-linearity. Specifically, the Walsh transform of a bent function in $m$ variables takes on exactly two values, $\pm 2^{m/2}$. Over the past four decades, these functions have attracted significant attention due to their crucial applications in coding theory, cryptography, design theory, and algebraic graph theory. Comprehensive surveys on bent functions are available in  \cite{C1,CM,M}.

Bending \cite{B}  presented the construction of two families of non-isomorphic 2-designs, known as addition designs and translation designs, using bent functions. Remarkably, these designs share the same parameters {when the dual of the corresponding bent function at zero vector is zero}. Additionally, Dempwolff and Neumann \cite{DN} extended Bending's work by constructing addition designs from $r$-plateaued functions, thereby generalizing the original construction. Furthermore, Polujan \cite{P} conducted a comprehensive review of Boolean and vectorial functions from a design-theoretic perspective, exploring various classes of cryptographic functions such as bent, plateaued, and differentially uniform functions, as well as the incidence structures derived from these mappings. In the work by Polujan and Pott \cite{PP}, Boolean and vectorial bent functions were utilized in the construction of addition designs and translation designs. Meidl, Polujan, and Pott \cite{MPP} further expanded on the applications of $(n,m)$-functions in constructing 2-designs, providing a new design-theoretic characterization of $(n,m)$-plateaued and $(n,m)$-bent functions, thereby enhancing the understanding of these functions.  The primary goal of this paper is to extend Bending's work from bent functions to plateaued functions.



 A further motivation for this work is to investigate the deep interconnections among combinatorial designs, linear codes, and Boolean functions, with the aim of resolving the two open problems described below.

Let $\C$ be an $[n,k,d]$ linear code over $\mathbb F_q$, and let $A_i=A_i(\C)$ be the number of codewords of Hamming weight $i$ in $\C ~(0\le i\le n)$. For each $j$ with $A_j\neq 0$, let $\mathcal B_j$ denote the set of the supports of all codewords of Hamming weight $j$ in $\C$, where the code coordinates are indexed by $1,2,\ldots, n$. Let $P=\{1,2,\ldots, n\}$. By providing sufficient conditions for the pair $(P, \mathcal B_k)$,   Ding, Tang, Wang, and their co-authors (see, e.g., \cite{DMT, DT0, DT1, DT, T, TD, WTD, XLW}) have made significant contributions in recent years to the construction of $t$-designs. Building on these results, Ding and Tang \cite{DT0} published a comprehensive monograph focused on $t$-designs derived from linear codes over finite fields. Of particular interest are two open problems posed in their monograph (see Open Problems \ref{111} and \ref{222} {with the code $\widetilde{\mathcal{C}}_{D_f}$ introduced in Section VI-B}), which concern the relationship between the equivalence of bent functions, the binary linear codes arising from bent functions, and the designs supported by these codes.


\begin{problem}{\rm \label{111}\cite[Problem 14.20]{DT0}}
Let $f$ and $g$ be two bent functions in $\mathcal{BF}_m$. What is the relationship between the equivalence of $f$ and $g$ and that of the codes $\widetilde{\mathcal{C}}_{D_f}$ and $\widetilde{\mathcal{C}}_{D_g}$?
\end{problem}

\begin{problem}{\rm \label{222} \cite[Problem 14.23]{DT0}}
Let $f$ and $g$ be two bent functions in $\mathcal{BF}_m$. What is the relationship between the equivalence of $f$ and $g$ and that of the designs with the same parameters held in the two codes $\widetilde{\mathcal{C}}_{D_f}$ and $\widetilde{\mathcal{C}}_{D_g}$?
\end{problem}



In this paper, we explore the deep interconnections among combinatorial designs, linear codes, and Boolean functions, with a particular focus on introducing the triple symmetric difference property (TSDP). The organization of the paper is as follows.

Section~\ref{s2} recalls basic notions related to designs, Boolean functions, and linear codes. 
Section~\ref{s3} presents a characterization of $2$-designs in terms of Boolean functions.
Section~\ref{s4} reviews known results on designs obtained from bent functions and offers a further extension along with an alternative proof.
Section~\ref{s5} begins with an alternative proof (Lemma~\ref{thm 2.23}) for addition designs derived from plateaued functions, which need not be simple or symmetric. We then characterize simple $2$-designs on $2^{m-r}$ points coming from $r$-plateaued functions in $m$ variables (Theorem~\ref{thm4.7}), paralleling Lemma~\ref{thm3.7}. A consequence (Corollary~\ref{cor2}) is that the addition designs derived from $r$-plateaued functions with no nonzero linear structure  satisfy the TSDP, giving rise to non-symmetric simple $2$-designs that possess the triple symmetric difference property but not the double symmetric difference property (Example~\ref{Example 4}).

Section~\ref{s6} consists of two parts. The first part extends several results (Theorem~\ref{thm4.8}, Proposition~\ref{propro1}, and Remark~\ref{rmk5}) to arbitrary $r$-plateaued functions with no nonzero linear structure. The second part first generalizes Open Problem~1 regarding affine equivalence to the setting of $r$-plateaued functions. Subsequently, through the construction of new quasi-symmetric designs (Theorem~\ref{theorembent}), we derive a unified result (Theorem~\ref{cor5}) analogous to Corollary~\ref{cor4}. This resolves Open Problems~1 and~2.

In Section~\ref{s6}, we determine the automorphism groups (Theorem~\ref{thm10}) of addition designs associated with $r$-plateaued functions and of their corresponding linear codes. 
Section~\ref{s7} uses Theorems~\ref{thm4.4} and~\ref{thm14} to construct new Boolean functions that yield two families of $2$-designs. These designs share the same parameters as the complement of a point-hyperplane design and addition designs arising from bent functions, respectively, yet they are non-isomorphic. 
Section~\ref{s8} concludes the paper.

\section{Preliminaries }\label{s2}
In this section, we will recall basic concepts on designs, linear codes, and Boolean functions. 

\subsection{Block designs}

 In this subsection, we recall some fundamental concepts of combinatorial designs and introduce a new notion: the triple symmetric difference property (TSDP) for arbitrary $2$-designs.

An incidence structure consists of a triple $\mathbb{D} = (P, \mathcal{B}, \mathcal{R})$, where $P$ and $\mathcal{B}$ are finite sets whose elements are called points and blocks, respectively, and $\mathcal{R} \subseteq P \times \mathcal{B}$ is a binary relation known as the incidence relation. If $\mathcal{R}$ is taken to be set inclusion, we simplify the notation and write $\mathbb{D} = (P, \mathcal{B})$. This paper focuses on 2-designs, which are formally defined as follows:

\begin{definition}\label{defn1}{\rm \cite{CD}
Let $v,k$ and $\lambda$ be positive integers. A 2-$(v,k,\lambda)$ design is a design $\mathbb{D}=(P,\mathcal{B})$ such that the following three conditions are satisfied:
\begin{itemize}
	\item [(i)] $|P|=v$;
	
	\item [(ii)] each block contains exactly $k$ points of $P$;
	
	\item [(iii)] every pair of distinct points is contained in exactly $\lambda$ blocks.
\end{itemize}
}
\end{definition}

By the replication number $r$, we mean that every point is contained in  exactly $r$ blocks. The complement of a design  $\mathbb{D}=(P,\mathcal{B})$ is the design  $\overline{\mathbb{D}}=(P,\overline{\mathcal{B}})$, where $\overline{\mathcal{B}}=\{P\setminus B:B\in\mathcal{B}\}$. The complement of a 2-$(v,k,\lambda)$ design is a 2-design with parameters $(v,v-k,b-2r+\lambda)$, where $b$ is the number of blocks and $r$ is the replication number.   We say that a design is simple if it contains no repeated blocks.

It is often convenient to represent a combinatorial design with an incidence matrix. A (point-block) incidence matrix $M$ of a design $(P, \mathcal{B})$
with $|P| = v$ and $|\mathcal{B}| = b$ is a $v\times b$ matrix whose rows and columns are respectively indexed by the elements $p$ of $P$ and the elements $B$ of $\mathcal{B}$ such that the entry in row $p$ and column $B$ is $1$ if the $p\in B$ and $0$ otherwise. The 2-rank is the rank over $\mathbb{F}_2$ of the point-block incidence matrix of a design. The binary linear code, denoted by $\mathcal{C}_{\mathbb{D}}$, of a design $\mathbb{D}$ is the $\mathbb{F}_2$-span of the rows of the incidence matrix \cite{B1, HP, L}.


Let $M$ and $N$ be $v\times b$ point-block incidence matrices of two designs $\mathbb{D}_1=(P,\mathcal{A})$ and $\mathbb{D}_2=(Q,\mathcal{B})$ respectively.
The designs $\mathbb{D}_1$ and $\mathbb{D}_2$ are isomorphic if there exist permutations $\sigma$ from $P$ to $Q$ and $\pi$ from $\mathcal{A}$ to $\mathcal{B}$ such that $M_{p,A}=N_{\sigma p,\pi A}$ for any $p\in P$ and $A\in\mathcal{A}$. There are a $v\times v$ permutation matrix $P_1$ and a $b\times b$ permutation matrix $P_2$ such that $M=P_1 N P_2$. The automorphism group, denoted $\Aut(\mathbb{D})$, of a design $\mathbb{D}=(P,\mathcal{B})$ with point-block incidence matrix $M$ is defined as the set of all ordered pairs $(\sigma,\pi)$ such that $\sigma$ and $\pi$ are permutations from $P$ to itself and from $\mathcal{B}$ to itself, respectively, and $M_{p,B}=M_{\sigma p,\pi B}$ for any $p\in P$ and $B\in\mathcal{B}$.
 Let $A$ and $B$ be two sets and  $A\Delta B:=(A\backslash B)\cup (B\backslash A)$  denote the  symmetric difference of  the two sets $ A$ and $B$.

\begin{definition}\label{defn2}{\rm \cite{CD, DS}
(1) A 2-$(v,k,\lambda)$ design is called {\it symmetric} if $v$ is equal to the number of blocks, equivalently, every two blocks share exactly $\lambda$ points.  

(2) A symmetric 2-design is said to have the {\it symmetric difference property}
or be an SDP design if given any three distinct blocks $B_1, B_2, B_3$ of the design, their symmetric difference $B_1\Delta B_2 \Delta B_3$ is either a block of the design or the complement of a block. 

(3) A 2-design is {\it quasi-symmetric} with intersection numbers $i$ and $j$, $(i < j)$ if any two blocks intersect in either $i$ or $j$ points.

}
\end{definition}

To extend the symmetric difference property to arbitrary 2-designs, next we introduce the concept of \lq\lq triple symmetric difference property'' as follows:

%

\begin{definition}\label{defn2@}{\rm 
  A 2-design is said to have the triple symmetric difference property or be a TSDP design if the symmetric difference $B_1 \Delta B_2 \Delta B_3$ is either a block of the design or the complement of a block for any three pairwise distinct blocks $B_1$, $B_2$, and $B_3$ of the design.


}
\end{definition}

In this paper, we employ TSDP to unify the research on symmetric difference properties of 2-designs.



\subsection{Boolean functions}

In this subsection, we introduce some fundamental concepts on Boolean functions, including the notion of plateaued functions, the equivalence of Boolean functions, and the automorphism group of a Boolean function, all of which play a central role in our study.

Let $\mathbb{F}_2$ be the finite field of size two and let $\mathbb{F}^m_2$ be a $\mathbb{F}_2$-vector space with dimension $m$. For a subset $C$ of $\mathbb{F}^m_2$, we denote by $C^*$ the set of non-zero elements of $C$ and by $\LI C\RI$ the linear span of $C$. A Boolean function $f$ in $m$ variables is a function from $\mathbb{F}^m_2$ to $\mathbb{F}_2$. We use the notation $\mathcal{BF}_m$ to denote the set of Boolean functions in $m$ variables. The degree of a Boolean function $f:\F^m_2\rightarrow \F_2$, denoted $\deg(f),$ is the highest number of variables appearing in any single term of its Algebraic Normal Form expressed by
$f(x_1,\ldots, x_m)=\sum_{I\subseteq\{1,\ldots,m\}}a_I\prod_{j \in I} x_j,$ $a_I\in\F_2$. 
A Boolean function $f$ is balanced if $|f^{-1}(0)|=|f^{-1}(1)|$.

Let $f\in\mathcal{BF}_m$ and let $P$ be a subset of $\mathbb{F}^m_2$.
The Walsh transform $W_f$ on $P$ of the $f$ is an integer-valued function defined as
\[
W_{P,f}(a)=\sum_{x\in P}(-1)^{f(x)+a\cdot x},
\]
where $\cdot$ is the dot product on $\mathbb{F}_{2}^{m}$, that is, $a \cdot x=a_{1}x_{1}+\cdots+a_{m}x_{m}$.
We denote by $W_f$ simply the Walsh transform $W_{\mathbb{F}^m_2,f}$ on $\mathbb{F}^m_2$ of $f$. Then
\[
(-1)^{f(a)}=\frac{1}{2^m}\sum_{x\in\mathbb{F}^m_2}W_f(x)(-1)^{a\cdot x},
\]
is  called the inversion formula for $f$. The Walsh support of $f$ in $\mathcal{BF}_m$ is defined as $S_f = \{a\in\F^m_2 : W_f(a)\neq 0\}$.



A Boolean function $f$ in $\mathcal{BF}_m$ is bent if $W_f(a)^2=2^m$ for any $a\in\mathbb{F}^m_2$ \cite{R}.
In this case, $m$ is even, and $W_f(a)$ can be written as $W_f(a)=(-1)^{f^*(a)}2^{\frac{m}{2}}$ for some Boolean function $f^*$ in $\mathcal{BF}_m$. We call $f^*$ the dual of $f$. It is known that $f^*$ is also bent and $f^{**}=f$. A bent function $f$ is self-dual if $f=f^*$. A well-known fact that the function $f(x,y)=x\cdot \sigma y+g(y)$ for all $x,y\in \mathbb{F}_2^m$ is  bent for any Boolean function $g$ in $\mathcal{BF}_m$ and permutation $\sigma$ of $\mathbb{F}^m_2$ \cite{M1}. These types of bent functions are referred to as Maiorana-McFarland (M-M) bent functions.

A bent function is also connected to Hadamard matrices in two ways: both $(f(x+y))_{x,y\in\mathbb{F}^m_2}$ and $(f(x)+f^*(y)+x\cdot y)_{x,y\in\mathbb{F}^m_2}$ are $(0,1)$-Hadamard matrices with constant row and column sums, known as regular Hadamard matrices. More generally, we define an $r$-plateaued function in $\mathcal{BF}_m$.
For an integer $r$ $(0\leq r\leq m)$, we say that a Boolean function $f$  in $\mathcal{BF}_m$ is $r$-plateaued if $W_f(a)\in\{0,\pm 2^{\frac{m+r}{2}}\}$ for all $a\in\mathbb{F}^m_2$. In this case, $m+r$ is even, and any $0$-plateaued function is just bent by Parseval's identity that $\sum_{x\in\mathbb{F}^m_2}W_f(x)^2=2^{2m}$. For the most recent survey on binary bent and plateaued functions, please see \cite[Chapter 6]{C3}.
 
The cross-correlation $C_{f,g}$ between two Boolean functions $f$ and $g$ in $\mathcal{BF}_m$ is defined as
\[
C_{f,g}(a)=\sum_{x\in\mathbb{F}^m_2}(-1)^{f(x)+g(x+a)}.
\] Then
\[
C_{f,g}(a)=\frac{1}{2^m}\sum_{x\in\mathbb{F}^m_2}W_f(x)W_g(x)(-1)^{a\cdot x}.
\]
The auto-correlation of $f$ is $C_{f,f}$, which is denoted by $C_f$ simply.  
It follows that $f$ is a bent function  in $\mathcal{BF}_m$ if and only if $C_f(a)=2^m\delta_{\mathbf{0},a}$ for all $a\in \mathbb F_2^m$, where $\delta$ is the Kronecker delta function. A nonzero vector $a\in\mathbb{F}^m_2$ is called a
linear structure of a Boolean function $f$ in $\mathcal{BF}_m$ if $f(x)+ f(x+a)$ is a constant for all $x\in\mathbb{F}^m_2$. So $f$  possesses a linear structure if and only if $|C_{f}(a)|=2^m$ for some $a\in\mathbb{F}_2^{m*}:=\mathbb{F}^m_2\setminus\{\bf{0}\}$, and bent functions do not have a linear structure.

\begin{definition}{\rm \cite{M} \label{affine}  Two Boolean functions $f$ and $g$ in $\mathcal{BF}_m$ are  (extended affine) equivalent (resp. affine equivalent) if there is an affine permutation $\sigma$ of $\mathbb{F}_2^m$ and an affine map $\pi$ from $\mathbb{F}_2^m$ to $\mathbb{F}_2$ such that $g(x)=f(\sigma x)+\pi x$ (resp. $g(x)=f(\sigma x)$), where $\sigma x=xA+a$ for $A$ is a non-singular matrix of size $m$, $a\in\mathbb{F}^m_2$ and $\pi x=b\cdot x+\varepsilon$ for $b\in\mathbb{F}^m_2$ and $\varepsilon\in\mathbb{F}_2$.

}
\end{definition}

\begin{definition} {\rm \cite{M} \label{ccz} Two Boolean functions $f$ and $g$ in $\mathcal{BF}_m$ are CCZ-equivalent if their graphs $\mathcal{G}_f=\{(x,f(x)):x\in\mathbb{F}^m_2\}$ and $\mathcal{G}_g=\{(x,g(x)):x\in\mathbb{F}^m_2\}$ are affine equivalent, that is, there is an affine permutation $\sigma$ of $\mathbb{F}^m_2\times\mathbb{F}_2$ such that $\sigma\mathcal{G}_f=\mathcal{G}_g$.

} 
\end{definition}

 There are implications of equivalences for Boolean functions \cite{BC1}: let $f$ and $g$ be Boolean functions. Then $f$ and $g$ are affine equivalent $\Longrightarrow$ $f$ and $g$ are (extended affine) equivalent $\Longleftrightarrow$ $f$ and $g$ are CCZ-equivalent. For further information and additional details, please refer to \cite{BC, BC1,CS, CM,EP,M}. Further, the degree of a Boolean function is invariant under affine and extended affine equivalences \cite{C3}. 

We denote by $\operatorname{GL}_m(\F_2)$ the group of invertible matrices of size $m$ with entries in $\mathbb{F}_2$ and by $A^T$ the transpose of a matrix $A$. For $[A,p,b,\varepsilon]\in \operatorname{GB}_m(\F_2):=\operatorname{GL}_m(\F_2)\times\F^m_2\times\F^m_2\times\F_2$ and $f\in\mathcal{BF}_m$, we define $[A,p,b,\varepsilon]f\in\mathcal{BF}_m$ by 
$
[A,p,b,\varepsilon]f(x)=f(xA+p)+x\cdot b+\varepsilon.
$
One can show that 
\[
[A,p,b,\varepsilon_1]\circ [B,q,b',\varepsilon_2]=[AB,pB+q,b+b'A^{T},p\cdot b'+\varepsilon_1+\varepsilon_2].
\]
This implies that
\[
\left\{\begin{pmatrix}
 1 & {\bf{0}} & {\bf{0}}\\
 b^T & A &{\bf{0}} \\
\varepsilon & p & {\bf{1}}
\end{pmatrix}\in\operatorname{GL}_{2m+1}(\F_2): [A,p,b,\varepsilon]\in\operatorname{GB}_m(\F_2) \right\}
\]
forms a subgroup of $\operatorname{GL}_{2m+1}(\F_2)$. The stabilizer $\operatorname{GB}_m(\F_2)_f=\{[A,p,b,\varepsilon]\in\operatorname{GB}_m(\F_2):[A,p,b,\varepsilon]f=f\}$ is the automorphism group of $f$. We denote this stabilizer group by $\Aut(f)$. 

\subsection{The linear code of a design}

In this subsection, we define an incidence structure $\mathbb{D} = (P, \{B^{f_b} : b \in B\})$ via Boolean functions $f_b$ labeled by $b \in B$, where each block consists of points $p \in P$ with $f_b(p) = 1$, and we then present the linear code constructed from this incidence structure.

Let $n$ and $k$ be positive integers. 
An $[n,k,d]$ linear code $\C$ over $\Bbb F_2$ is a $k$-dimensional subspace of $\Bbb F_2^n$ with minimum Hamming distance $d$. The Hamming weight $w(c)$ of a codeword $c$ in $\C$ is the number of nonzero coordinate positions. The support of a vector $x$ in $\mathbb{F}^m_2$ is the set of non-zero coordinate positions of $x$. The weight distribution of an $[n,k,d]$ code $\C$ is defined as $1+A_1z+\cdots+A_nz^n$, where $A_i$ is the number of codewords in $\C$ of Hamming weight $i$. The dual code of an $[n,k,d]$ linear code $ \C$ over $\Bbb F_2$ is defined as $
\C^{\bot} = \{ x\in \Bbb F_2^{n} : x\cdot y=0\text{ for all } y\in \C\}.$
If $\mathcal{C}^{\bot} \subseteq\mathcal{C} $, then $\mathcal{C}$ is called a {\it dual-containing} code.
If $\mathcal{C} \subseteq\mathcal{C}^{\bot}$, then $\mathcal{C}$ is called a {\it self-orthogonal} code.

The character sum $\chi_a(X)$ for $a\in\mathbb{F}^m_2$ and $X\subseteq\mathbb{F}^m_2$ is defined as 
\[
\chi_a(X)=\sum_{x\in X}(-1)^{a\cdot x}.
\]
It is known that if $\C$ is linear, then $\chi_a(\C)=|\C|1_{\C^{\perp}}(a)$ for $a\in\F^m_2$, where $1_X$ stands for the indicator function of a subset $X$ of $\F^m_2$.
Two binary linear codes $\C_1$ and $\C_2$ are equivalent if there is a permutation matrix $P$ such that the linear span of a matrix $G_1$ is $\C_1$ if and only if the linear span of $G_1P$ is $\C_2$. The automorphism group of a binary linear $\C$, denoted $\Aut(\C),$ is the group of all permutations of the coordinates of 
$\C$ that preserve the code \cite{MS1}.

Let $P$ and $B$ be two subsets of $\mathbb{F}_2^m$ and $f_b$'s be Boolean functions in $\mathcal{BF}_m$ labeled by  $b\in B$, and let $\mathcal{B}$ be a collection of $B^{f_b}$ defined by
\begin{align}\label{eqn:(1)}
	B^{f_b}=\{p\in P: f_b(p)=1\}.
\end{align}
Then $\mathbb{D}=(P,\mathcal{B})$ is an  incidence structure. Notice that in the complement of $\mathbb{D}$, its blocks are
$\{p\in P: f_b(p)=0\}$.
The point-block incidence matrix of $\mathbb{D}$ can be written as 
$
M(\mathbb{D})=(f_b(p))_{p\in P, b\in B}.$
We call $\mathbb{D}=(P,\{B^{f_b}:b\in B\})$ an incidence structure associated with Boolean functions $f_b$ in $\mathcal{BF}_m$ corresponding to $b\in B$. We say that a pair $(P,\{f_b\in\mathcal{BF}_m:b\in B\})$ induces a 2-design  if $\mathbb{D}=(P,\{B^{f_b}:b\in B\})$ is a 2-design. 

Recall that the binary linear code $\mathcal{C}_{\mathbb{D}}$ of a design $\mathbb{D}$ is the $\mathbb{F}_2$-span of the rows of the incidence matrix. Similarly, one can define the incidence matrix of an incidence structure $\mathbb{D}$ and its binary linear code $\C_{\mathbb{D}}$.
The Hamming weight of $(f_b(p))_{b\in B}$ is 
\[
w((f_b(p))_{b\in B})=|\{b \in B: p \in B^{f_{b}}\}|.
\]
Since $f_b$ is not affine in general, it is much more difficult to compute the weight distribution of $\mathcal{C}_{\mathbb{D}}$ of $\mathbb{D}$.

\subsection{Two lemmas}

The following lemmas are important and will be used later.

\begin{lemma}\label{lem3.6}{\rm
Let $X$ be a subset of $\mathbb{F}^k_2$ and let $G$ be a $k\times m$ matrix over $\mathbb{F}_2$. Then
\[
\LI\{uG:u\in X\}\RI=\{uG:u\in \LI X\RI\}.
\]
}
\end{lemma}
\begin{proof}
    The proof is straightforward.
\end{proof}

\begin{lemma}\label{lem2}{\rm
   Let $f$ be a non-affine function in ${\mathcal{BF}_m}$ and $\sigma$ be a permutation of $\mathbb{F}^m_2$. Then $f \circ \sigma^{-1}$ is non-affine if and only if the $\mathbb{F}_2$-linear span of 
   $\{(f(x),\sigma x,1):x\in\mathbb{F}^m_2\}$ is $\mathbb{F}_2\times\mathbb{F}^m_2\times\mathbb{F}_2$.   }    
   
\end{lemma}

\begin{proof}
      Let \[
 G=\begin{pmatrix}
\cdots & 1 & \cdots\\
\cdots & (\sigma x)^T & \cdots \\
\cdots & f(x) & \cdots
\end{pmatrix}_{x\in\mathbb{F}^m_2}
\]We only need to prove that the rank of $G$ is $m+2$ if and only if $f \circ \sigma^{-1}$ is non-affine.
For any $u \in \mathbb{F}_{2}^{m}, t, s \in \mathbb{F}_{2}$, the Hamming weight of the vector $c(t, u, s)=(t, u, s)G$ is given as follows:
\begin{equation*}
    \begin{split}
        w(c(t, u, s))&=2^{m}-\sum_{x \in \mathbb{F}_{2}^{m}}\delta_{0, t+u \cdot \sigma x+sf(x)}\\
        &=2^{m}-\frac{1}{2}\sum_{x \in \mathbb{F}_{2}^{m}}(1+(-1)^{t+u \cdot \sigma x+sf(x)})\\
        &=
\begin{cases}
    2^{m-1}-\frac{1}{2}(-1)^{t}W_{f\circ \sigma^{-1}}(u), & \text{ if } s=1,\\
    2^{m-1}-(-1)^{t}2^{m-1}\delta_{{\bf{0}}, u}, & \text{ if } s=0.
\end{cases}
    \end{split}
\end{equation*}
Note that if $w(c(t, u, 0))=0$, then $u={\bf{0}},t=0$. If there is $(u_{0}, t_{0}) \in \mathbb{F}_{2}^{m} \times \mathbb{F}_{2}$ such that $w(c(t_{0}, u_{0}, 1))=0$, then $|W_{f\circ \sigma^{-1}}(u_{0})|=2^m$. In this case, the Parseval's identity guarantees that $W_{f\circ \sigma^{-1}}(x)=0$ for any $x \neq u_{0}$, and then $f\circ \sigma^{-1}$ is an affine function. Therefore, if $f \circ \sigma^{-1}$ is non-affine, then $w(c(t, u, s))=0$ if and only if $(t, u, s)=(0,{\bf{0}},0)$, and the rank of $G$ is $m+2$. Conversely, if $f \circ \sigma^{-1}$ is affine, then there is $(a, b) \in \mathbb{F}_{2}^{m} \times \mathbb{F}_{2}$ such that $f \circ \sigma^{-1}(x)=a \cdot x+b$, and $w(c(b, a, 1))=0$ since $W_{f \circ \sigma^{-1}}(a)=2^m(-1)^{b}$. Thus, if $f \circ \sigma^{-1}$ is affine, then the rank of $G$ is not $m+2$. This completes the proof.
\end{proof}

\section{Characterization of 2-designs in terms of Boolean functions}\label{s3}

In this section, we characterize 2-designs based on Eq. \eqref{eqn:(1)} in term of Boolean functions.

\begin{theorem}\label{lem3.1}{\rm
For two subsets $P,B\subseteq \mathbb{F}^m_2$,
let $\mathbb{D}=(P,\{B^{f_b}:b\in B\})$ be an incidence structure associated with Boolean functions $f_b$ in $\mathcal{BF}_m$ corresponding to $b\in B$. The following statements are true.
\begin{itemize}
\item[(i)] The sum $\chi_u(B^{f_b})
=\frac{1}{2}\left(\chi_u(P)-W_{P,f_b}(u)\right)$, \text{ and the size of a block equals } 
$$\frac{1}{2}\left(|P|-W_{P,f_b}(\mathbf{0})\right).$$

\item[(ii)] The number of blocks containing a point $p\in P$ equals
$\frac{1}{2}\left(|B|-\sum_{b\in B}(-1)^{f_b(p)}\right).$

\item[(iii)] The number of blocks containing two distinct points $p,q\in P$ equals
$$\frac{1}{4}\left(|B|-\sum_{b\in B}(-1)^{f_b(p)}
 -\sum_{b\in B}(-1)^{f_b(q)}+\sum_{b\in B}(-1)^{f_b(p)+f_b(q)}\right).$$
 
\item[(iv)] The block intersection numbers of two distinct vectors $a,b\in B$ are
$$\frac{1}{4}\left(|P|-\sum_{p\in P}(-1)^{f_a(p)}
 -\sum_{p\in P}(-1)^{f_b(p)}+\sum_{p\in P}(-1)^{f_a(p)+f_b(p)}\right).$$
\end{itemize}
}
\end{theorem}

\begin{proof}
(i). We have that for $u\in\mathbb{F}^m_2$ and $b\in B$,
\begin{eqnarray*}
&\chi_u(B^{f_b})=\sum_{x\in B^{f_b}}(-1)^{u\cdot x}=\sum_{x\in P}(-1)^{u\cdot x}\delta_{0,f_b(x)+1}\\
&=\frac{1}{2}\sum_{x\in P}(-1)^{u\cdot x}(1+(-1)^{f_b(x)+1}),
\end{eqnarray*}
and the first assertion follows. The second assertion follows from that $\chi_{\mathbf{0}}(B^{f_b})=|B^{f_b}|$.

(ii). We have
\[
 |\{b\in B:p\in B^{f_b}\}|
 =\sum_{b\in B}\delta_{0,f_b(p)+1}
 =\frac{1}{2}\sum_{b\in B}(1+(-1)^{f_b(p)+1}),
\]
where $\delta$ is the Kronecker delta function, and the assertion follows.

(iii). We have that for two distinct points $p,q\in P$, 
\begin{multline*}
 |\{b\in B:p,q\in B^{f_b}\}|
 =\sum_{b\in B}\delta_{0,f_b(p)+1}\delta_{0,f_b(q)+1}\\
 =\frac{1}{4}\sum_{b\in B}(1+(-1)^{f_b(p)+1})(1+(-1)^{f_b(q)+1})\\
 =\frac{1}{4}\sum_{b\in B}\left(1+(-1)^{f_b(p)+1}+(-1)^{f_b(q)+1}+(-1)^{f_b(p)+f_b(q)}\right),
\end{multline*}
and the assertion follows. 

(iv) We have
\begin{multline*}
  |B^{f_a}\cap B^{f_b}|
  =\sum_{p\in P}\delta_{0,f_a(p)+1}\delta_{0,f_b(p)+1}\\
 =\frac{1}{4}\sum_{p\in P}(1+(-1)^{f_a(p)+1})(1+(-1)^{f_b(p)+1})\\
 =\frac{1}{4}\sum_{p\in P}\left(1+(-1)^{f_a(p)+1}+(-1)^{f_b(p)+1}+(-1)^{f_a(p)+f_b(p)}\right),
\end{multline*}
and the assertion follows. 
\end{proof}

By Theorem 1 and Definition 1,   we have the following corollary.

\begin{corollary}\label{cor3.2}{\rm
Let $P$ and $B$ be subsets of $\mathbb{F}^m_2$. A pair $(P,\{f_b\in\mathcal{BF}_m:b\in B\})$ induces a $2$-$(v, k, \lambda)$ design $\mathbb{D}=(P,\{B^{f_b}:b\in B\})$ if and only if the following three conditions hold;
\begin{itemize}
    \item [(i)] $W_{P,f_b}(\mathbf{0})=v-2k$ for all $b\in B$.

\item [(ii)] $\sum_{b\in B}(-1)^{f_b(p)}=|B|-2r_p$ is a constant for all $p\in P$, where $r_p$'s are replication numbers of $\mathbb{D}$.

\item[(iii)] $\sum_{b\in B}(-1)^{f_b(p)+f_b(q)}=4(\lambda-r_p)+|B|$ is a constant for all $p,q\in P$ with $p\neq q$.
\end{itemize}
}
\end{corollary}


A linear code with a complementary dual (or an LCD code) is defined to be a linear code $\C$ whose dual code $\C^{\perp}$ satisfies $\C\cap \C^{\perp}=\{\mathbf{0}\}$. See \cite{Ma} for more information.

\begin{proposition}\label{exam3.3} 
{\rm
Let $\C$ be an LCD linear code of length $m$, and let $f_b(x)=b\cdot x$ for $b\in \C^*$ be a Boolean function in $\mathcal{BF}_m$. Then
 the pair $(\C^*,\{f_b\in\mathcal{BF}_m:b\in \C^*\})$ induces a symmetric 2-design with parameters $(|\C|-1,|\C|/2,|\C|/4)$. 
 
 } 
 
\end{proposition}

\begin{proof}
Let $f_b(x)=b\cdot x$ for $b\in \C^*$. We will check the three conditions of Corollary \ref{cor3.2}. 
First, we have that for any $b\in \C^*$,
\[
W_{\C^*,f_b}(\mathbf{0})=\sum_{p\in C^*}(-1)^{b\cdot p}=-1+|\C|1_{\C^{\perp}}(b)=-1,
\]
which is a constant.
Next, we have that for any $p\in \C^*$,
\[
\sum_{b\in \C^*}(-1)^{f_b(p)}=\sum_{{b}\in \C^*}(-1)^{b\cdot p}=-1,
\]
which is a constant. Finally, we have that for two distinct points $p,q\in \C^*$,
\[
\sum_{b\in \C^*}(-1)^{f_b(p)+f_b(q)}=\sum_{b\in \C^*}(-1)^{(p+q)\cdot b}
=-1+|\C|1_{\C^{\perp}}(p+q)=-1,
\]
which is a constant. The parameters are computed by using  Corollary \ref{cor3.2}.
\end{proof}

 \begin{remark}\label{rmk3.4}
 	{\rm
 	 Note that this symmetric $2$-design in Proposition \ref{exam3.3} is the complement of a point-hyperplane design which is already known in \cite[Theorem 6.20]{CD}.

  }
 \end{remark}

\section{Designs from Bent functions}   \label{AT}




In this section, we summarize existing results on designs derived from bent functions. We also present a further extension (Proposition \ref{thm3.8}) on quadratic bent functions and an alternative proof (Theorem \ref{thmB3}) for the classification of minimal rank translation designs of a special class of M-M bent functions.

\subsection{Some well-known results on designs derived from bent functions}

Bending \cite{B} constructed two families of symmetric 2-designs with the same parameters: so called the addition designs and the translation designs. We discuss these two constructions in the following remark.  

Let $f$ be a bent function in $\mathcal{BF}_m$ and let $f^*$ be its dual.

\begin{itemize}
\item  \textbf{Addition designs:} 
 For each $b\in \mathbb{F}^m_2$, let $f_b(x)=b\cdot x+f^*(b)+f(x)$ for $x\in \mathbb{F}^m_2$. A pair $(\mathbb{F}^m_2,\{f_b\in\mathcal{BF}_m:b\in \mathbb{F}^m_2\})$ induces a symmetric $2$-design with parameters 
 \[
(2^m,2^{m-1}-2^{\frac{m}{2}-1},2^{m-2}-2^{\frac{m}{2}-1}).~~~ (\dag)
\]
We call $\mathbb{AD}_f$ the addition design of a bent function $f$, and this concept will be generalized to any $r$-plateaued function with no nonzero linear structure.

\item  \textbf{Translation designs:} For each $b\in \mathbb{F}^m_2$, let $f_b(x)=f(x+b)$. A pair $(\mathbb{F}^m_2,\{f_b\in\mathcal{BF}_m:b\in \mathbb{F}^m_2\})$ induces a symmetric $2$-design with parameters $$(2^m,2^{m-1}-(-1)^{f^*(\mathbf{0})}2^{\frac{m}{2}-1},2^{m-2}{-}(-1)^
 {f^*(\mathbf{0})}2^{\frac{m}{2}-1}).$$ 
We call $\mathbb{TD}_f$ the translation design of a bent function $f$.

\end{itemize}

A family of symmetric designs with parameters $(\dag)$
was first introduced by Block~\cite{B1} referred to as a symplectic design by Kantor \cite{K} later, and ~\cite{CS}. Kantor \cite{K} conducted further studies on symplectic designs and discovered their inherent triple symmetric difference property (TSDP), and computed the automorphism groups of symplectic designs as well as proved that any TSDP design must have parameters of this form. In \cite{BM}, Bracken and McGuire characterized the TSDP designs generating four-weight spin models with two values, showing that symplectic TSDP designs are the only possible ones.

Notice that any symmetric 2-design with parameters $(\dag)$ has $2$-rank at least $m+2$  \cite{MW}, and minimal 2-rank designs with parameters $(\dag)$ have been characterized as follows.

\begin{lemma}\label{thm3.7}  {\rm $($\cite{B,K,DS,MW}$)$
Let $\mathbb{D}$ be a symmetric 2-design with $2^m$ points. The following statements are equivalent.
\begin{itemize}
    \item[(i)] $\mathbb{D}$ is a TSDP design with parameters $(\dag)$.

    \item[(ii)] $\mathbb{D}$ is isomorphic to some addition design derived from a bent function in $\mathcal{BF}_m$.
    
    \item[(iii)] The $2$-rank of $\mathbb{D}$ with parameters $(\dag)$ is $m+2$.

   \item[(iv)] The binary linear code $\mathcal{C}_{\mathbb{D}}$ of $\mathbb{D}$ is generated by the first order Reed-Muller code with parameters $[2^m,m,2^{m-1}]$ and the characteristic vector of a difference set.
\end{itemize}
}
\end{lemma}

In the following, we prove that the sufficient condition in \cite[Theorem 11.9]{B}  is  also necessary.


\begin{proposition}
    \label{thm3.8}{\rm  

Let $f$ be a bent function in $\mathcal{BF}_m$  with $f^*(\mathbf{0})=0$. The following statements are equivalent.
\begin{itemize}
    \item[(i)] $f$ is quadratic.
    \item[(ii)] The addition design $\mathbb{AD}_f$ and translation design $\mathbb{TD}_f$ derived from $f$ are isomorphic with an identity map on point sets and some linear permutation map on block sets.
    \end{itemize}
    }
 \end{proposition}

\begin{proof}
(i) $\Longrightarrow$ (ii): It was verified in \cite[Theorem 11.9]{B}.

(ii) $\Longrightarrow$ (i): Assume that $f(p)+f^*(\pi b)+p\cdot\pi b =f(p+b)$ for some linear permutation $\pi$ of $\mathbb{F}_2^m$. By plugging $p=\mathbf{0}$, we have $f^*(\pi b)=f(b)+f(\mathbf{0})$ for all $b\in\mathbb{F}^m_2$, and so  
\begin{align}\label{eqn:(6)}
    f(p)+f(b)+f(p+b)=p\cdot\pi b+f(\mathbf{0})
\end{align}
for all $p,b\in\mathbb{F}^m_2$. Let $g$ be the function of $x,y\in \mathbb{F}_2^m$ defined by
\begin{align}\label{eqn:(7)}
	g(x,y)=f(x)+f(y)+f(x+y)+x\cdot\pi y+f(\mathbf{0}).
\end{align}
Then $g(p,b)=0$ for all $p,b\in\mathbb{F}^m_2$ by Eq. (\ref{eqn:(6)}), and so the Boolean $g$ is identically zero, that is, $g\equiv 0$. Notice that the degree of $f(x)+f(y)+f(x+y)$ is equal to the degree of $f$. Observe that the degree of $x\cdot\pi y$ is two, since $\pi$ is linear. Hence $f$ is quadratic using Eq. (\ref{eqn:(7)}). This completes the proof.
\end{proof}

To illustrate our Proposition  \ref{thm3.8} we give the following example.
\begin{example}\label{ex2.9}
	
	{\rm
	{Let $f(x_1,x_2,x_3,y_1,y_2,y_3)=x_1y_1+x_2y_2+x_3y_3$ be an M-M bent function in $\mathcal{BF}_6$. Let $\pi$ be the permutation of $\mathbb{F}_2^6$ defined as $\pi(x_1,x_2,x_3,y_1,y_2,y_3)=(y_1,y_2,y_3,x_1,x_2,x_3)$. One can check that $f(x)+f^{*}(\pi b)+x\cdot\pi b=f(x+b),$}
	{where $x=(x_1,x_2,x_3,y_1,y_2,y_3)$ and $b=(a_1,a_2,a_3,b_1,b_2,b_3)$ in $\F^6_2$. Thus the addition design and translation design derived from $f$ are isomorphic with an identity map on point sets and $\pi$ being a linear permutation on block sets.}
 }
\end{example}

\subsection{ Classification of minimal rank translation designs of a special class of M-M bent functions }

 Weng, Feng, and Qiu \cite{WFQ} investigated the ranks of translation designs derived from various classes of bent functions. Recall that any symmetric 2-design with parameters $(\dag)$ has 2-rank at least $m + 2$. In this subsection, we are interested in classifying bent functions $f$ in $m$ variables for which $\mathbb{TD}_f$ has the {TSDP}, that is, $\mathbb{TD}_f$ has minimal rank $m+2$.  We classify the minimal rank translation designs of a special class of M-M bent functions (Theorem \ref{thmB3}).



\begin{lemma}\label{lem.SDP}{\rm
	{Let $P$ and $B$ be subsets of $\mathbb{F}^m_2$. 
 A 
 2-design induced the blocks $B^{f_b}$ for $b\in\mathbb{F}^m_2$ has the TSDP if and only if for any $b_1,b_2,b_3\in B$, there is a $b_4\in B$ such that $f_{b_1}(x)+f_{b_2}(x)+f_{b_3}(x)=f_{b_4}(x)+\varepsilon_{b_4}$ for all $x\in P,$ where $\varepsilon_{b_4}$ is a constant in $\mathbb{F}_2$.}
	}
\end{lemma}
\begin{proof}
	{The result follows by noting that $x\in B^{f_{b_1}}\Delta B^{f_{b_2}}\Delta B^{f_{b_3}}$ if and only if $x$ is in exactly one of $B^{f_{b_1}}$, $B^{f_{b_2}}$ and $B^{f_{b_3}}$ or in all of three of them, which is $f_{b_1}(x)+f_{b_2}(x)+f_{b_3}(x)=1$.}
\end{proof}


\begin{theorem}\label{thmB3}{\rm
{Let $f(x,y)=x\cdot y+ g(y)$ in $\mathcal{BF}_{2m}$ be an M-M bent function. Then the translation design $\mathbb{TD}_f$ derived from $f$ has the TSDP, {which is equivalent to having minimal rank $2m+2$}, if and only if the degree of $g$ is less than or equal to three.

 } 
}
\end{theorem}

\begin{proof}
	From this point onward, we represent $f_1(x)\fallingdotseq f_2(x)$ if $f_1(x)=f_2(x)+s$ for all $x\in\mathbb{F}_2^n$ and a  constant $s\in\mathbb{F}_2$.
	
	{($\Leftarrow$):  Let $f(x,y)=x\cdot y+g(y)$ with deg $g \le 3$.}
		{By Lemma \ref{lem.SDP}, we have to show that for any $a_1,a_2,a_3,b_1,b_2,b_3\in\mathbb{F}_2^m$, there are $a_4,b_4\in\mathbb{F}_2^m$ such that $f(x+a_1,y+b_1)+f(x+a_2,y+b_2)+f(x+a_3,y+b_3)\fallingdotseq f(x+a_4,y+b_4)$. Observe that $x+a_1$ and $y+b_1$ can be treated as $x$ and $y$, respectively.  It is thus enough to show that for any $a,b,c,d\in\mathbb{F}_2^m$, there are $\alpha,\beta\in\mathbb{F}_2^m$ such that}
	\begin{align}\label{eqn:SDP}
		&f(x,y)+f(x+a,y+b)+f(x+c,y+d)\nonumber\\
		&\fallingdotseq f(x+\alpha,y+\beta)
	\end{align}
{for all $ x,y\in \mathbb{F}_2^m$.}
{Since $f(x,y)=x\cdot y +g(y)$, then}
\begin{multline*}
	f(x,y)+f(x+a,y+b)+f(x+c,y+d)\\
	\fallingdotseq (x+a+c)\cdot(y+b+d)+g(y)+g(y+b)+g(y+d).
\end{multline*}
{If $g$ is a constant or linear, then $g(y)+g(y+b)+g(y+d)\fallingdotseq g(y+b+d)$, and Eq. (\ref{eqn:SDP}) holds when $\alpha=a+c$ and $\beta=b+d$.}
{Next we consider the quadratic Boolean function $h_1(x_1,x_2)=x_1 x_2$. Then}
\begin{align*}
&h_1(x_1,x_2)+h_1(x_1+a_1,x_2+a_2)+h_1(x_1+b_1,x_2+b_2)\\ &=x_1x_2+(a_2+b_2)x_1+(a_1+b_1)x_2+a_1a_2+b_1b_2\\ &\fallingdotseq h_1(x_1+a_1+b_1,x_2+a_2+b_2).
\end{align*}
{If $g$ is a quadratic Boolean function, then $g(y)+g(y+b)+g(y+d)\fallingdotseq g(y+b+d)$ since $g$ is a linear combination of monomials whose degree is less than or equal to two. Therefore Eq. (\ref{eqn:SDP}) holds when $\alpha=a+c$ and $\beta=b+d$.}
Next we consider the cubic Boolean function $h_2(x_1,x_2,x_3)=x_1x_2x_3$. Then,
\begin{align*}
	&h_2(x_1,x_2,x_3)+h_2(x_1+a_1,x_2+a_2,x_3+a_3)\\
	&+h_2(x_1+b_1,x_2+b_2,x_3+b_3)\\
	&=x_1x_2x_3+(a_3+b_3)x_1x_2+(a_2+b_2)x_1x_3+(a_1+b_1)x_2x_3\\
	&\quad+(a_2a_3+b_2b_3)x_1+(a_1a_3+b_1b_3)x_2+(a_1a_2+b_1b_2)x_3\\
	&+a_1a_2a_3+b_1b_2b_3\\
	&\fallingdotseq h_2(x_1+a_1+b_1,x_2+a_2+b_2,x_3+a_3+b_3)\\
	&\quad +(a_2b_3+a_3b_2)x_1+(a_1b_3+a_3b_1)x_2+(a_1b_2+a_2b_1)x_3.
\end{align*}
Thus if $g$ is a cubic Boolean function, then
\begin{align*}
	g(y)+g(y+b)+g(y+d)\fallingdotseq g(y+b+d)+\sum_{i=1}^{m}\gamma_i y_i
\end{align*}
where $\gamma_i\in\mathbb{F}_2$ are constants which depend on $b$ and $d$, and $y\in\mathbb{F}_2^m$.
That is, 
\begin{align*}
	&f(x,y)+f(x+a,y+b)+f(x+c,y+d)\\
	&\fallingdotseq (x+a+c)\cdot(y+b+d)+g(y+b+d)+\sum_{i=1}^{m}\gamma_i y_i\\
	&\fallingdotseq (x+a+c+\gamma)\cdot(y+b+d)+g(y+b+d)\\
	&=f(x+a+c+\gamma,y+b+d)
\end{align*}
Therefore, $f(x,y)=x\cdot y+g(y)$ induces a {TSDP} design when deg $g \le 3$.

	{($\Rightarrow$): Let $f(x,y)=x\cdot y+g(y)$ be an M-M bent function with deg $g \ge 4$. Assume that $f$ induces a {TSDP} design. { By Lemma \ref{lem.SDP}}, for all $a,b,c,d\in\mathbb{F}_2^m$, there exist $\alpha,\beta\in\mathbb{F}_2^m$ such that}
\begin{align}\label{eqn:SDP2}
	&f(x,y)+f(x+a,y+b)+f(x+c,y+d)\nonumber\\
	&\fallingdotseq f(x+\alpha,y+\beta).
\end{align}
If $y={\bf 0}$ in Eq. (\ref{eqn:SDP2}), 
\begin{align*}
	&f(x,0)+f(x+a,b)+f(x+c,d)\\
	&=(x+a)\cdot b+(x+c)\cdot d +g(0)+g(b)+g(d)\fallingdotseq x\cdot (b+d).
\end{align*}
Since $f(x+\alpha,\beta)=(x+\alpha)\cdot \beta+g(\beta)\fallingdotseq x\cdot \beta$, we obtain that $x\cdot (b+d)\fallingdotseq x\cdot \beta$. So it follows that $\beta=b+d$.
If $x={\bf 0}$ in Eq. (\ref{eqn:SDP2}),
\begin{align*}
	&f(0,y)+f(a,y+b)+f(c,y+d)\\
	&=a\cdot(y+b)+c\cdot(y+d)+g(y)+g(y+b)+g(y+d)\\
	&\fallingdotseq(a+c)\cdot y +g(y)+g(y+b)+g(y+d)
\end{align*}
Note that $f(\alpha,y+\beta)=\alpha\cdot(y+\beta)+g(y+\beta)\fallingdotseq \alpha\cdot y+g(y+\beta)$. Since $\beta=b+d$, we have
\begin{align*}
	&(a+c)\cdot y+g(y)+g(y+b)+g(y+d)\\
	&\fallingdotseq \alpha \cdot y+g(y+b+d),
\end{align*}
which implies that
\begin{align}{\label{eqn:SDP3}}
	&g(y)+g(y+b)+g(y+d)+g(y+b+d)\nonumber\\
	&\fallingdotseq (a+c+\alpha)\cdot y.
\end{align}
Let $g$ be a Boolean function of degree $k\geq 4$ that contains at least one monomial $h_1(y_1,\ldots,y_m):=\prod_{i=1}^{k}y_i$ of degree $k$. Then the degree of $h_1(y)+h_1(y+b)+h_1(y+d)+h_1(y+b+d)$ is $k-2$, which implies that the degree of $g(y)+g(y+b)+g(y+d)+g(y+b+d)$ in Eq. \eqref{eqn:SDP3}  is at most $k-2$. 

{Now let $g(x)=h(x)+r(x),$ where $h(x)$ is a Boolean function with degree $k$ terms and each term involves $\prod_{i=1}^{k-2}y_i$, and $r(x)$ is remained terms of $g(x)$. Then $h(x)=\prod_{i=1}^{k-2}y_i\sum_{(i,j)\in I}y_iy_j$ for some subset $I\subset \{1,\ldots, m\}\times\{1,\ldots, m\}$. Note that $I$ is nonempty, since $(k-1,k-2)\in I$.}

{It is obvious that $h_1(y)+h_1(y+b)+h_1(y+d)+h_1(y+b+d)$ consists of a $y_1y_2\cdots y_{k-2}$ term with coefficient $b_{k-1}d_k+b_kd_{k-1}$. Then the coefficient of $\prod_{i=1}^{k-2}y_i$ in $g(y)+g(y+b)+g(y+d)+g(y+b+d)$, which is the same as the coefficient of $\prod_{i=1}^{k-2}y_i$ in $h(y)+h(y+b)+h(y+d)+h(y+b+d)$, is $\sum_{(i,j)\in I}{b_id_j+b_jd_i}$. Therefore the coefficient of $\prod_{i=1}^{k-2}y_i$ is not identically zero, and it implies that the degree of $g(y)+g(y+b)+g(y+d)+g(y+b+d)$ is $k-2$ which is at least $2$. This is a contradiction to (\ref{eqn:SDP3}) and completes the proof.}  
\end{proof}

\begin{remark}
{\rm During the revision of this manuscript, the authors observed that the rank of $f(x,y) = x \cdot y + g(y)$ in Theorem 2 had already been established in \cite[Theorem 3.5]{WFQZ}. Thus, Theorem 2 can be derived from \cite[Theorem 3.5]{WFQZ} and Lemma \ref{thm3.7}. However, as our proof differs entirely from that in \cite{WFQZ}, we have chosen to keep it.
}
 
\end{remark}

We finish  this {section} with the following problems:
\begin{problem}\label{openproblemnew} 
{\rm 
{ Classify the minimal rank translation designs derived from bent functions.}
}
\end{problem}




\section{ Designs derived from plateaued functions}\label{s4}

Recall that bent functions give rise to symmetric simple $2$-designs with identical parameters: addition designs (which satisfy the TSDP) and translation designs (which, in general, do not). This naturally raises the question for $r$-plateaued functions. To address this, in the first subsection we provide an alternative proof (Lemma \ref{thm 2.23}) for addition designs $\mathbb{AD}_f$ (not necessarily simple or symmetric) derived from $r$-plateaued functions $f$ in $m$ variables, and we show (Proposition \ref{prop2}) that plateaued functions of degree at most one induce precisely translation designs. In the second subsection, we focus on characterizing simple $2$-designs with $2^{m-r}$ points derived from $r$-plateaued functions (Theorem \ref{thm4.7}), in analogy with Lemma \ref{thm3.7}. As a byproduct, we prove (Corollary \ref{cor2}) that the $\mathbb{AD}_f$ satisfies the TSDP, producing non-symmetric simple 2-designs that satisfy the triple symmetric difference property while failing the double symmetric difference property (Example \ref{Example 4} right above).

\subsection{Addition designs from plateaued functions: a new proof}

In \cite{DN}, Dempwolff and Neumann constructed  the addition designs from $r$-plateaued functions, which is a generalization of the Bending's work presented. In this subsection, we provide  a very simple proof for this in our framework. 

\begin{lemma}{\rm  \label{thm 2.23} \cite{DN} 
For an integer $r$ $(0\leq r < m)$, let $f$ be an $r$-plateaued function in $\mathcal{BF}_m,$ and let $P=S_f$ be the Walsh support of $f$. We can write $W_f(x)=(-1)^{g(x)}2^{\frac{m+r}{2}}1_P(x)$ for some $g$ in $\mathcal{BF}_m$. For each $b\in \mathbb{F}^m_2$, let $f_b(x)=b\cdot x+f(b)+g(x)$. Then a pair $(P,\{f_b\in\mathcal{BF}_m:b\in\mathbb{F}^m_2\})$ induces a 2-design $\mathbb{AD}_f=(P,\{B^{f_b}:b\in \mathbb{F}^m_2\})$ with parameters $(2^{m-r},2^{m-r-1}-2^{\frac{m-r-2}{2}},2^{m-2}-2^{\frac{{m+r-2}}{2}})$. The intersection numbers are $2^{m-r-2}-2^{\frac{m-r-2}{2}}+2^{-2-r}(-1)^{f(a)+f(b)}C_f(a+b)$
for $a,b\in\mathbb{F}^m_2$. If $r\neq0$, then $\mathbb{AD}_f$ is a non-symmetric 2-design. We call $\mathbb{AD}_f$ the addition design of $f$.
}
\end{lemma}
\begin{proof}
In view of the identity that $\sum_{x\in\mathbb{F}^m_2}W_f(x)^2=2^{2m}$, we have $|P|=2^{m-r}$. 
By the inversion formula for $f$, we have 
\[
2^m(-1)^{f(b)}=\sum_{p\in \mathbb{F}^m_2}W_f(p)(-1)^{b\cdot p}=2^{\frac{m+r}{2}}\sum_{p\in P}(-1)^{g(p)+b\cdot p}.
\]
It follows that for $b\in\mathbb{F}^m_2$,
\begin{eqnarray*}
&W_{P,f_b}(\mathbf{0})=\sum_{p\in P}(-1)^{b\cdot p+f(b)+g(p)}\\
&=(-1)^{f(b)}2^{\frac{m-r}{2}}(-1)^{f(b)}=2^{\frac{m-r}{2}},
\end{eqnarray*}
for $p\in P$,
\[
\sum_{b\in\mathbb{F}^m_2}(-1)^{f_b(p)}=\sum_{b\in \mathbb{F}^m_2}(-1)^{b\cdot p+f(b)+g(p)}=2^{\frac{m+r}{2}}
\]
and for two distinct points $p,q\in P$,
\[
\sum_{b\in \mathbb{F}^m_2}(-1)^{f_b(p)+f_b(q)}=(-1)^{g(p)+g(q)}\sum_{b\in \mathbb{F}^m_2}(-1)^{(p+q)\cdot b}=0.
\]
The parameters are computed by those three values of summations using  Theorem \ref{lem3.1}.
\end{proof}

We are interested in finding the translation design corresponding to an $r$-plateaued function not being bent. 
\begin{proposition}\label{prop2}{\rm 
    Let $f$ be an $r$-plateaued function in $\mathcal{BF}_m$ which is not bent and let for each $b \in \mathbb{F}_{2}^{m}$, $f_{b}(x)=f(x+b) \text{ for all } x \in \mathbb{F}_{2}^{m}$. Then the third condition in Corollary \ref{cor3.2} holds if and only if either $f$ is an affine function in  $\mathcal{BF}_1$ or $f$ is a constant function. 
    }
   \end{proposition}
\begin{proof}
    The third condition in Corollary \ref{cor3.2} is equivalent to that $|\{x \in \mathbb{F}_{2}^{m}: f(x)+f(x+a)=0\}|$ is a constant for any $a \neq \bf{0}$. Assume that $|\{x \in \mathbb{F}_{2}^{m}: f(x)+f(x+a)=0\}|$ is a constant for any $a \neq\bf{0}$ and denoted the constant by $\ell$. Then 
\begin{equation*}
\begin{split}
& (2^m-1)\ell+2^m=\sum_{a \in \mathbb{F}_{2}^{m}}|\{x \in \mathbb{F}_{2}^{m}: f(x)+f(x+a)=0\}|\\
 &=\frac{1}{2}\sum_{a, x \in \mathbb{F}_{2}^{m}}\sum_{z \in \mathbb{F}_{2}}(-1)^{(f(x)+f(x+a))z}\\
 &=\frac{1}{2}\sum_{x \in \mathbb{F}_{2}^{m}}\sum_{z \in \mathbb{F}_{2}}(-1)^{zf(x)}\sum_{a \in \mathbb{F}_{2}^{m}}(-1)^{zf(x+a)}\\
 &=\frac{1}{2}\sum_{z \in \mathbb{F}_{2}}\sum_{x \in \mathbb{F}_{2}^{m}}(-1)^{zf(x)}\sum_{a \in \mathbb{F}_{2}^{m}}(-1)^{zf(a)}\\
 &=2^{2m-1}+\frac{1}{2}W_{f}({\bf{0}})^{2},
\end{split}
\end{equation*}
which implies that $\ell=\frac{2^{2m-1}-2^m+\frac{1}{2}W_{f}({\bf{0}})^{2}}{2^m-1}$. If ${\bf{0}} \notin S_{f}$, then $2^m-1 \mid 2^{m}(2^{m-1}-1)$, which implies that $m=r=1$, and so $f(x)=x+s$ for $x,s \in \mathbb{F}_{2}$. If ${\bf{0}} \in S_{f}$ (that is, $W_{f}({\bf{0}})^{2}=2^{m+r}$), then $2^m-1 \mid 2^m(2^{m-1}+2^{r-1}-1)$, which implies that $r=m$, and so $f$ is a constant function as ${\bf{0}} \in S_{f}$. This completes  the proof.
\end{proof}



Theorem 3.5 in \cite{DN} shows that if $\mathbb{AD}_f$ derived from a $1$-plateaued function is not simple, then it has block multiplicity $2$. We are interested in $r$-plateaued functions $f$ for which $\mathbb{AD}_f$ is simple.  

\begin{lemma}\label{lem4.3}
 {\rm For an integer $r$ $(0\leq r\leq m)$, let $f$ be an $r$-plateaued function in $\mathcal{BF}_m$ and let $S_f$ be represented as $S_f=a+E$ for some subset $E$ of $\mathbb F_2^m$ and $a\in S_f$. Then the following statements are equivalent.
\begin{itemize}
    \item [(i)] The addition design  $\mathbb{AD}_f$ is simple.

     \item [(ii)] $f$ has no nonzero linear structure.

      \item [(iii)] The maximal number of linearly independent vectors in $E$ is equal to $m$. 

      \item [(iv)] A generator matrix of the linear code $\mathcal{C}_{\mathbb{AD}_{f}}$ of the design $\mathbb{AD}_{f}$ is $G_{f}$, where 
      \[
    G_f=\begin{pmatrix}
   \cdots & 1 & \cdots\\
   \cdots & x^T & \cdots \\
   \cdots & f(x) & \cdots
  \end{pmatrix}_{x\in\mathbb{F}^m_2}.\] 
\end{itemize}
}   
\end{lemma}
\begin{proof}
    (i) $\Longleftrightarrow$ (ii): {A design is not simple if and only if there is exactly one intersection number which is equal to size of a block $k$. In Lemma \ref{thm 2.23}, the design is not simple if and only if \begin{eqnarray*}&2^{m-r-1}-2^{\frac{m-r-2}{2}}\\
    &=2^{m-r-2}-2^{\frac{m-r-2}{2}}+2^{-2-r}(-1)^{f(a)+f(b)}C_f(a+b),\end{eqnarray*} for some distinct $a,b\in\mathbb{F}_2^m$. That is,
	$2^m=(-1)^{f(a)+f(b)}C_f(a+b)$ and $|C_f(a+b)|=2^m$ for some distinct $a,b\in\mathbb{F}_2^m$ which implies that $f$ has a nonzero linear structure. Therefore, the design in Lemma \ref{thm 2.23} is simple if and only if $f$ has no nonzero linear structure.}

    (ii) $\Longleftrightarrow$ (iii): See  \cite{HP}.

   (iii) $\Longleftrightarrow$ (iv): Note that (iii) and (iv) both imply that $f$ is non-affine, namely $r<m$. Then the rank of $G_{f}$ is $m+2$ whenever (iii) or (iv) holds.
    For each $p\in S_{f}$, let $f_b(p)=b\cdot p+g(p)+f(b)$ for any $b\in\mathbb{F}^m_2$ by our assumption. Notice that
	${(f_b(p))_{b\in \mathbb{F}_{2}^{m}}}=(g(p),p,1)G_f$. By Lemma \ref{lem3.6}, the binary linear code $\mathcal{C}_{\mathbb{AD}_{f}}$ is $\LI\{(g(p), p, 1)G_{f}: p \in S_{f}\}\RI
 =\{c(t, u, s)=(t, u, s)G_{f}: (t, u, s) \in \LI\{(g(p), p, 1): p \in S_{f}\}\RI\}$. Define 
 \[
G_{g}=\begin{pmatrix}
\cdots & 1 & \cdots\\
\cdots & x^T & \cdots \\
\cdots & g(x) & \cdots
\end{pmatrix}_{x\in S_{f}}.
\]
Therefore, (iv) holds if and only if the rank of $G_{g}$ is $m+2$.
The Hamming weight of the vector $d(t, u, s)=(t, u, s)G_{g}$ is given as follows:
\begin{multline*}
 w(d(t,u,s))=2^{m-r}-\sum_{x\in S_{f}}\delta_{0,t+u\cdot x+sg(x)}\\
 =2^{m-r}-\frac{1}{2}\sum_{x\in S_{f}}(1+(-1)^{t+u\cdot x+sg(x)})
  \\=\begin{cases}
  2^{m-r-1}-\frac{1}{2}(-1)^{t}\sum_{x \in S_{f}}(-1)^{g(x)+u \cdot x} & \text{ if } s=1,\\
  2^{m-r-1}-\frac{1}{2}(-1)^t\sum_{x \in S_{f}} (-1)^{u \cdot x} & \text{ if } s=0.
 \end{cases}
\end{multline*}
By the inversion formula for $f$, $\sum_{x \in S_{f}}(-1)^{g(x)+u \cdot x}=2^{\frac{m-r}{2}}(-1)^{f(u)}$, and then $w(d(t, u, 1)) \neq 0$ by $r<m$. By $\textbf{0} \in E$ and $|E|=2^{m-r}$, $w(d(t, u, s))=0$ if and only if $s=0, t=u \cdot a, u \cdot x=0$ for all $x \in E$. Thus the rank of $G_{g}$ is $m+2$ if and only if (iii) holds.
\end{proof}

\begin{remark}{\rm  \label{remarkff}
Let $\mathcal{C}_{f}$ be the linear code generated by $G_{f}$, that is, $\mathcal{C}_{f}=\{c(t, u, s)=(t, u, s)G_{f}: (t, u, s) \in \mathbb{F}_{2} \times \mathbb{F}_{2}^{m} \times \mathbb{F}_{2}\}$. By Lemma \ref{lem4.3} and its proof, if $\mathbb{AD}_{f}$ is not simple, that is, $f$ has nonzero linear structure, then $\mathcal{C}_{\mathbb{AD}_{f}}$ is a subcode of $\mathcal{C}_{f}$.}
\end{remark}


\subsection{Designs with $2^{m-r}$ points }

In analogy with the construction of symmetric $2$-designs on $2^m$ points from bent functions, this subsection is devoted to characterizing simple $2$-designs on $2^{m-r}$ points arising from $r$-plateaued functions (Theorem \ref{thm4.7}). Applying this characterization and Lemma \ref{thm 2.23},  we obtain non-symmetric addition designs derived from $r$-plateaued functions that possess the TSDP (Corollary \ref{cor2}).


   





\begin{theorem}\label{thm4.7}{\rm 

     
Let $\mathbb{D}$ be a simple $2$-design with $2^{m-r}$ points for some integer $0 \leq r < m$. The following statements are equivalent.
\begin{itemize}
    \item[(i)] $\mathbb{D}$ has the TSDP with parameters 
    $2$-$(2^{m-r},2^{m-r-1}-2^{\frac{m-r-2}{2}},2^{m-2}-2^{\frac{{m+r-2}}{2}})$.

    \item[(ii)] $\mathbb{D}$ is isomorphic to some addition design derived from an $r$-plateaued function in $\mathcal{BF}_m$ with no nonzero linear structure.
    
    \item[(iii)] The 2-rank of 
    $\mathbb{D}$ with parameters  $2$-$(2^{m-r},2^{m-r-1}-2^{\frac{m-r-2}{2}},2^{m-2}-2^{\frac{{m+r-2}}{2}})$ is $m+2$, namely, $\operatorname{rank}_2(M(\mathbb{D}))=m+2$, and the first order Reed-Muller code $\mathrm{RM}(1,m)$ is a subcode of $\mathcal{C}_{\mathbb{D}}$.
\end{itemize}
}
\end{theorem}
\begin{proof}
	(i) $\Longrightarrow$ (ii): 
Let $\mathbb{D}=({P},\mathcal{B})$ be a 2-design such that ${P}=\{0,1,\ldots,2^{m-r}-1\}$ ($0\leq r<m$) is a point set of $\mathbb{D}$ and $\mathcal{B}$ is a collection of block sets satisfying the TSDP with $|\mathcal{B}| = 2^m$. To avoid confusion in this proof, we let $p$ denote a point in ${P}$. Since $\mathbb{D}$ satisfies the TSDP, by the proof of \cite[Theorem 9.11]{B}, for a fixed block $B_{0} \in \mathcal{B}$, we can label all blocks by vectors in $\mathbb{F}_2^m$ such that for all $x, y \in \mathbb{F}_2^m$, {up to complementation,}
\begin{equation}\label{labelB}
B_{0} \triangle B_x \triangle B_y \triangleq B_{x+y},
\end{equation}
and then we label each point {$p\in {P}$ 
with a vector $z = \sum_{i=1}^{{m}} z_i e_i \in \mathbb{F}_2^m$, where $\{e_{1}, \dots, e_{m}\}$ is a standard basis of $\mathbb{F}_{2}^{m}$ (that is, $e_{1}=(1, 0, \dots, 0, 0), \dots, e_{m}=(0, 0, \dots, 0, 1)$) and 
$$z_i = [{p} \in B_{0} \triangle B_{e_i}] \oplus [0 \in B_{0} \triangle B_{e_i}],$$
where $[\cdot]$ is $1$ if the condition holds and $0$ otherwise, and $\oplus$ is addition mod $2$. Thus there is an injection $\varphi: P \to \mathbb{F}_2^m$ and for each point $p\in P$, 
\begin{equation}\label{labelP}
\varphi(p)= \sum_{i=1}^{m} z_i e_i \in \mathbb{F}_2^m.
\end{equation}

For the point $0 \in {P}$,  we define a Boolean function $f: \mathbb{F}_2^m \to \mathbb{F}_2$ by $f(x)=[0 \in B_{x}]$.
From the labeling procedure, we observe that 
for each nonzero vector $x \in \mathbb{F}_2^m$, the value $\varphi(p) \cdot x$ is a constant 
for all $p \in B_{0} \triangle B_x$. Indeed, {by \eqref{labelP}, for any nonzero $x \in \mathbb{F}_{2}^{m}$ (denoting $x=e_{i_{1}}+\dots+e_{i_{t}}$) and $p \in {P}$, we have $\varphi(p) \cdot x=\sum_{j=1}^{t}([p \in B_{0} \triangle B_{e_{i_{j}}}]\oplus [0 \in B_{0} \triangle B_{e_{i_{j}}}])$. From \eqref{labelB}, the value of $\varphi(p) \cdot x$ are determined by $[p \in B_{0} \triangle B_{x}]$ and $[0 \in B_{0} \triangle B_{x}]$.}
{ Using this observation and defining $f_{b}(\varphi(p))=[p \in B_{b}],$ we have that $f_{b}(\varphi(p))+\varphi(p) \cdot b+f(b)$ depends only on $p\in {P}$. Denote $g(\varphi(p))=f_{b}(\varphi(p))+\varphi(p) \cdot b+f(b)$.}
{ To improve readability in the proof, let $\varphi(P)=\{\overline p\in\mathbb{F}^m_2:p\in P \}$.} Then the incidence matrix of $\mathbb{D}$ is represented as $(f_b(\overline p))_{\overline p \in \varphi(P), b \in \mathbb{F}_2^m}$, 
where
\[
f_b(\overline p) = b \cdot {\overline p} + f(b) + g(\overline p)
\]
for some Boolean function $g: \varphi({P})\to \mathbb{F}_2$. The parameters of the design implies that the replication number is $2^{\frac{m+r}{2}}$ and $W_f(x)=(-1)^{g(x)}2^{\frac{m+r}{2}}$ for all $x\in {\varphi({P})}$. Note that $|{P}|=2^{m-r}$ and the Parseval's identity guarantees that $W_f(x)=0$  for $x\in\mathbb{F}_2^m\backslash {\varphi({P})}$. Thus $f$ is an $r$-plateaued function with $S_{f}=\varphi({P})$. Thus $\mathbb{D}$ is isomorphic to a design derived from $f$ by Lemma \ref{thm 2.23}. Since $\mathbb{D}$ is simple, $f$ has no nonzero linear structure by Lemma \ref{lem4.3}.


	(ii) $\Longrightarrow$ (i):   By assumption, we have $\mathbb{D}=\mathbb{AD}_f$ for some plateaued function $f$ with no nonzero linear structure. 
    By Lemma \ref{thm 2.23}, we also have $f_b(p)=b\cdot p+f(b)+g(p)$. Let us denote by $D_f$ the $f^{-1}(1)$ for $f\in \mathcal{BF}_m$. For $a\in\mathbb{F}^m_2$, we define $\ell_a(x)=a\cdot x$ for all $x\in\mathbb{F}^m_2$. Then for $b\in\mathbb{F}^m_2$, $B^{f_b}$ equals $D_{g+\ell_b}$ if $f(b)=0$ and $D_{g+\ell_b+1}$ otherwise. Set $h_{b,\varepsilon}=g+\ell_b+\varepsilon$ for some constant $\varepsilon\in\mathbb{F}^m_2$. Then for each $b, b', b''\in\mathbb{F}_2^m$ and $\varepsilon=\varepsilon_1+\varepsilon_2+\varepsilon_3$, we have
         \begin{eqnarray*}
    &&D_{h_{b,\varepsilon_1}}\Delta D_{h_{b',\varepsilon_2}}\Delta D_{h_{b'',\varepsilon_3}}\\
       & =&\{p\in\mathbb{F}^m_2:(h_{b,\varepsilon_1}(p),h_{b',\varepsilon_2}(p),h_{b'',\varepsilon_3}(p))\in\{(1,0,0),(0,1,0),(0,0,1),(1,1,1)\}\}\\
       & =&\{p\in\mathbb{F}^m_2:h_{b,\varepsilon_1}(p)+h_{b',\varepsilon_2}(p)+h_{b'',\varepsilon_3}(p)=1\}\\&=&\{p\in\mathbb{F}^m_2:h_{b+b'+b'',\varepsilon}(p)=1\}=D_{b+b'+b''+\varepsilon},   
        \end{eqnarray*}
      which is a block labeled by $b+b'+b''$ or its complement {by assuming that the three blocks $D_{h_{b,\varepsilon_i}}$ for $b\in\mathbb{F}^m_2$ and $i=1,2,3$ are  distinct.} This prove that  $\mathbb{D}=\mathbb{D}_f$ has the TSDP. }

(ii) $\Longrightarrow$ (iii): It follows from Lemma \ref{lem4.3}.
 (iii) $\Longrightarrow$ (ii): Since $\mathbb{D}$ is a $2$-design with parameters $(2^{m-r},2^{m-r-1}-2^{\frac{m-r-2}{2}},2^{m-2}-2^{\frac{{m+r-2}}{2}})$, then the number of blocks is \begin{eqnarray*}&&\lambda\frac{v(v-1)}{k(k-1)}\\
 &&=(2^{m-2}-2^{\frac{{m+r-2}}{2}})\frac{2^{m-r}(2^{m-r}-1)}{(2^{m-r-1}-2^{\frac{m-r-2}{2}})(2^{m-r-1}-2^{\frac{m-r-2}{2}}-1)}\\
 &&=2^m.\end{eqnarray*} Denote $M(\mathbb{D})=(m_{i, j})_{1 \leq i \leq 2^{m-r}, 1 \leq j \leq 2^m}$. Let $P$ be a subset of $\mathbb{F}_{2}^{m}$ with size $2^{m-r}$ and $P=\{p_{1}, \dots, p_{2^{m-r}}\}$, and let $\mathbb{F}_{2}^{m}=\{b_{1}, \dots, b_{2^m}\}$. Define $f_{b_{j}}(p_{i})=m_{i, j}, 1 \leq i \leq 2^{m-r}, 1 \leq j \leq 2^m$. By Theorem \ref{lem3.1}, 
    \begin{align}
        & \sum_{p \in P}(-1)^{f_{b}(p)}=2^{\frac{m-r}{2}} \text{ for any } b \in \mathbb{F}_{2}^{m},  \label{2}\\ 
        & \sum_{b \in \mathbb{F}_{2}^{m}}(-1)^{f_{b}(p)}=2^{\frac{m+r}{2}} \text{ for any } p \in P, \label{3}\\
        & \sum_{b \in \mathbb{F}_{2}^{m}}(-1)^{f_{b}(p)+f_{b}(q)}=0 \text{ for any distinct } p, q \in P \label{4}.
    \end{align}
Note that a generator matrix of  $\mathrm{RM}(1, m)$ is
\begin{equation*}
    \begin{pmatrix}
		\cdots & 1 & \cdots\\
		\cdots & b^T & \cdots \\
	\end{pmatrix}_{b\in \mathbb{F}_{2}^{m}}
\end{equation*}
As the dimension of $\mathcal{C}_{\mathbb{D}}$ is $m+2$, we have a codeword $\alpha=(\alpha_{1}, \dots, \alpha_{2^m}) \in \mathcal{C}_{\mathbb{D}} \backslash \mathrm{RM}(1, m)$. Define $f(b_{j})=\alpha_{j}, 1 \leq j \leq 2^{m}$. In the following, we prove that $f$ is $r$-plateaued. By the above arguments, a generator matrix of $\mathcal{C}_{\mathbb{D}}$ is
\begin{equation*}
    \begin{pmatrix}
		\cdots & 1 & \cdots\\
		\cdots & b^T & \cdots \\
		\cdots & f(b) & \cdots
	\end{pmatrix}_{b\in \mathbb{F}_{2}^{m}}.
\end{equation*}
Note that $(f_{b}(p))_{b \in \mathbb{F}_{2}^{m}}$ is a codeword. Then there is unique $(t(p), u(p), s(p)) \in \mathbb{F}_{2} \times \mathbb{F}_{2}^{m} \times \mathbb{F}_{2}$ such that $f_{b}(p)=t(p)+u(p)\cdot b+s(p)f(b)$. If $s(p)=0$, then $\sum_{b \in \mathbb{F}_{2}^{m}}(-1)^{f_{b}(p)}=\sum_{b \in \mathbb{F}_{2}^{m}}(-1)^{t(p)+u(p)\cdot b}=(-1)^{t(p)}\delta_{\textbf{0}}(u(p))2^m$, which contradicts Eq. \eqref{3}. Therefore, for any $p \in P, b \in \mathbb{F}_{2}^{m}$, 
\begin{equation} \label{5}
    f_{b}(p)=t(p)+u(p) \cdot b+f(b).
\end{equation}
By Eqs. \eqref{4} and \eqref{5}, for any distinct $p, q \in P$, we have
\begin{eqnarray*}&0=\sum_{b \in \mathbb{F}_{2}^{m}}(-1)^{f_{b}(p)+f_{b}(q)}\\
&=\sum_{b \in \mathbb{F}_{2}^{m}}(-1)^{t(p)+t(q)+(u(p)+u(q)) \cdot b}\\
&=(-1)^{t(p)+t(q)}\delta_{\textbf{0}}(u(p)+u(q))2^m,\end{eqnarray*} which implies that $u(p) \neq u(q)$ for any distinct $p, q \in P$. Then the size of $S\triangleq\{u(p): p \in P\}$ is $2^{m-r}$. By Eqs. \eqref{3} and \eqref{5}, we have \begin{eqnarray*}&2^{\frac{m+r}{2}}=\sum_{b \in \mathbb{F}_{2}^{m}}(-1)^{f_{b}(p)}
=\sum_{b \in \mathbb{F}_{2}^{m}}(-1)^{t(p)+u(p)\cdot b+f(b)}\\
&=(-1)^{t(p)}\sum_{b \in \mathbb{F}_{2}^{m}}(-1)^{f(b)+u(p) \cdot b},\end{eqnarray*} which implies that $|W_{f}(u(p))|=2^{\frac{m+r}{2}}$ for any $u(p) \in S$. The Parseval's identity guarantees that $W_{f}(a)=0$ for any $a \in \mathbb{F}_{2}^{m} \backslash S$. Thus $f$ is $r$-plateaued with $S_{f}=S$ and $t(p)=g(u(p))$, where $W_{f}(x)=2^{\frac{m+r}{2}}(-1)^{g(x)}1_{S_{f}}(x)$. Then by Eq. \eqref{5},
$f_{b}(p)=f(b)+g(u(p))+u(p) \cdot b.$
Hence $\mathbb{D}$ is isomorphic to the addition design $\mathbb{AD}_{f}$ derived from $f$. 
 By Lemma \ref{lem4.3}, $f$ has no nonzero linear structure.

 This completes the proof.
\end{proof}

In the following, we construct the non-symmetric 2-designs satisfying the TSDP, which follow from Lemma \ref{thm 2.23} and Theorem \ref{thm4.7}.

\begin{corollary}\label{cor2} {\rm 
 For an integer $r$ $(0\leq r < m)$, let $f$ be an $r$-plateaued function in $\mathcal{BF}_m$ with no nonzero linear structure. Then the following statements are true.
    \begin{itemize}
    \item[(i)] $\mathbb{AD}_f$ has the TSDP with parameters 
    $2$-$(2^{m-r},2^{m-r-1}-2^{\frac{m-r-2}{2}},2^{m-2}-2^{\frac{{m+r-2}}{2}})$.

    
    \item[(ii)] The 2-rank of 
    $\mathbb{AD}_f$ with parameters  $2$-$(2^{m-r},2^{m-r-1}-2^{\frac{m-r-2}{2}},2^{m-2}-2^{\frac{{m+r-2}}{2}})$ is $m+2$, and the first order Reed-Muller code $\mathrm{RM}(1,m)$ is a subcode of $\mathcal{C}_{\mathbb{AD}_f}$. 
\end{itemize}
}
\end{corollary}





In \cite{CD}, a 2-design is said to have the double  symmetric difference property or be a DSDP design if the symmetric difference $B_1 \Delta B_2$ is either a block of the design or the complement of a block for any pair of distinct blocks $B_1$ and $B_2$ of the design.
We provide an example of a non-symmetric TSDP design constructed from Corollary \ref{cor2} that does not satisfy the DSDP defined in \cite{CD}.

\begin{example}\label{Example 4}
{\rm
Let us consider $f(x_{1}, \dots, x_{5})=x_{1}x_{3}+x_{2}x_{4}+x_{1}x_{2}x_{5}$ which is given in \cite[Example 4.1]{HP}. Then $f$ is a $1$-plateaued function with no nonzero linear structure, $S_{f}=\{(x_{1}, x_{2}, x_{3}, x_{4}, x_{3}x_{4}): x_{1}, x_{2}, x_{3}, x_{4} \in \mathbb{F}_{2}\}$ and $g(x_{1}, \dots, x_{5})=x_{1}x_{3}+x_{2}x_{4}$ for $(x_{1}, \dots, x_{5}) \in S_{f}$.  For $b=(b_{1}, \dots, b_{5}) \in \mathbb{F}_{2}^{5}$ and $p=(p_{1}, \dots, p_{5}) \in P$, we have $P=S_{f}$ and $f_{b}(p)=f(b)+g(p)+b \cdot p
=b_{1}b_{3}+b_{2}b_{4}+b_{1}b_{2}b_{5}+p_{1}p_{3}+p_{2}p_{4}+b_{1}p_{1}+b_{2}p_{2}+b_{3}p_{3}+b_{4}p_{4}+b_{5}p_{3}p_{4}$. One can verify that $\mathbb{AD}_f$ is a $2$-$(16,6,4)$ design satisfying the TSDP. Further, we have \begin{eqnarray*}&B^{f_{(0,0,0,0,0)}}\Delta B^{f_{(0,0,0,0,1)}}\\
&=\{(0,0,1,1,1),(0,1,1,1,1),(1,0,1,1,1),(1,1,1,1,1)\},\end{eqnarray*} which is not a block or the complement of a block. Thus $\mathbb{AD}_f$ does not satisfy the DSDP.
}
\end{example}

For the case $r=0$,  the condition (iii) in Theorem \ref{thm4.7} can be simplified to ``the 2-rank of $\mathbb{D}$ with parameters  $2$-$(2^{m-r},2^{m-r-1}-2^{\frac{m-r-2}{2}},2^{m-2}-2^{\frac{{m+r-2}}{2}})$ is $m+2$" by Lemma \ref{thm3.7}.  We left it as an open problem for the case $r \geq 1$.

\begin{problem}{\rm   \label{openpro3} 
When $r \geq 1$, simplify the condition (iii) in Theorem \ref{thm4.7} as in Lemma \ref{thm3.7}.
}
    
\end{problem}

\section{Designs, linear codes, plateaued functions and their links}\label{s5}

Equivalent relationships between addition designs, their linear codes, the designs held in those linear codes, and bent functions have been established in \cite{B, DS, EP, DMT, DN} (Corollary \ref{cor4}). In the first subsection, we partially generalize these results (Theorem \ref{thm4.8}, Proposition \ref{propro1}, and Remark \ref{rmk5}) to any $r$-plateaued function with no nonzero linear structure. 
In the second subsection, we first generalize Open Problem 1 concerning affine equivalence to $r$-plateaued functions, and then by constructing new quasi-symmetric designs (Theorem \ref{theorembent}), we also obtain a unified result (Theorem \ref{cor5}) that is analogous to Corollary \ref{cor4}. As a byproduct, we settle two Open Problems 1 and 2.


\subsection{On the addition design $\mathbb{AD}_f$ derived from an $r$-plateaued function}

In this subsection, we consider the  equivalent relationships between addition designs of $r$-plateaued functions, linear codes of the addition designs, and $r$-plateaued functions.

\begin{theorem}\label{thm4.8}

	{\rm Let $r$ be an integer with $0\leq r<m$. Let $f$ and $g$ be $r$-plateaued functions in $\mathcal{BF}_m$ with no nonzero  linear structures. Then the following  statements are equivalent.
 \begin{itemize}
     \item[(i)] $f$ and $g$ are equivalent.
     
     \item[(ii)] $\mathbb{AD}_f$ and $\mathbb{AD}_g$ defined in  Lemma \ref{thm 2.23} are isomorphic.
     
     \item[(iii)] $\mathcal{C}_{\mathbb{AD}_f}$ and $\mathcal{C}_{\mathbb{AD}_g}$ are equivalent.

     

 \end{itemize}
 }
\end{theorem}



\begin{proof}


    (i) $\Longrightarrow$ (ii): We can write $W_f(x)=(-1)^{\tilde{f}(x)}2^{\frac{m+r}{2}}$ for $x\in S_f$ and $W_g(x)=(-1)^{\tilde{g}(x)}2^{\frac{m+r}{2}}$ for $x\in S_g$.
    Since $f$ and $g$ are equivalent, there is an affine permutation $\sigma$ of $\mathbb{F}_{2}^{m}$ and $y \in \mathbb{F}_{2}^{m}, s \in \mathbb{F}_{2}$ such that
    $
        f(x)=g(\sigma x)+y \cdot x+s \text{ for all }x \in \mathbb{F}_{2}^{m}.
   $
    For $p \in S_{f}$, we have
    \begin{eqnarray*}&W_{f}(p)=\sum_{x \in \mathbb{F}_{2}^{m}}(-1)^{f(x)+p \cdot x}\\
    &
            =(-1)^{s}\sum_{x \in \mathbb{F}_{2}^{m}}(-1)^{g(\sigma x)+(y+p) \cdot x}\\
            &
            =(-1)^{s}\sum_{x \in \mathbb{F}_{2}^{m}}(-1)^{g(x)+(y+p) \cdot \sigma^{-1} x}.
           \end{eqnarray*}
 Here $\sigma^{-1}$ is  an affine permutation of $\mathbb{F}_{2}^{m}$. Denote $\sigma^{-1} x=\tau x+\gamma$, where $\tau$ is a linear permutation of $\mathbb{F}_{2}^{m}$ and $\gamma \in \mathbb{F}_{2}^{m}$. 
Let $\tau^{*}$ \footnote{Since $\tau$ is a linear function of $\mathbb{F}_{2}^{m}$, then  $y \cdot \tau(x)$ is a linear Boolean function for any $y \in \mathbb{F}_{2}^{m}$. Hence, there is $\tau^{*}(y) \in \mathbb{F}_{2}^{m}$ such that $y \cdot \tau(x)=\tau^{*}(y) \cdot x$. If $\tau^{*}(y)=\tau^{*}(y')$, then $y \cdot \tau(x)=y' \cdot \tau(x)$ for all $x \in \mathbb{F}_{2}^{m}$. As $\tau$ is a permutation, we obtain $y=y'$, and thus $\tau^{*}$ is a permutation of $\mathbb{F}_{2}^{m}$.} be the adjoint permutation of $\tau$, that is, $y \cdot \tau(x)=\tau^{*}(y) \cdot x$.  Then
    \begin{eqnarray*}
         &W_{f}(p)=(-1)^{s+(y+p) \cdot \gamma} \sum_{x \in \mathbb{F}_{2}^{m}}(-1)^{g(x)+\tau^{*}(y+p) \cdot x}\\
         &=(-1)^{s+(y+p) \cdot \gamma}W_{g}(\tau^{*}(y+p)),
    \end{eqnarray*}
    which implies that $\tau^{*}(y+p) \in S_{g}$ and
    $
        \widetilde{f}(p)=\widetilde{g}(\tau^{*}(y+p))+(y+p) \cdot \gamma+s.
    $
    Then for $p \in S_{f}$ and $b \in \mathbb{F}_{2}^{m}$, we have 
    \begin{equation*}
        f(b)+ \widetilde{f}(p)+b \cdot p=g(\sigma b)+ \widetilde{g}(\tau^{*}(y+p))+(y+p) \cdot (b+\gamma).
    \end{equation*}
    Since $\sigma(b) \cdot \tau^{*}(y+p)=\tau(\sigma b) \cdot (y+p)=(b+\gamma) \cdot (y+p)$, 
     \begin{equation*}
        f(b)+ \widetilde{f}(p)+b \cdot p=g(\sigma b)+ \widetilde{g}(\tau^{*}(y+p))+\sigma(b) \cdot \tau^{*}(y+p),
    \end{equation*}
    which implies that $\mathbb{AD}_f$ and $\mathbb{AD}_g$ are isomorphic.

(ii) $\Longrightarrow$ (iii):
Let $M_{f}$ and $M_{g}$ be the incidence matrices of the designs $\mathbb{AD}_f$ and  $\mathbb{AD}_g$, respectively. 
By the definition of a code equivalence, we have that 
\begin{equation*}
    \begin{split}
        & \mathbb{AD}_f \text{ and }\mathbb{AD}_g \text{ are isomorphic } \\
         &\Longrightarrow M_{g}=PM_{f}Q \text{ for some permutation matrices } P \text{ and }Q\\
         & \Longrightarrow \mathcal{C}_{\mathbb{AD}_g}=\mathcal{C}_{\mathbb{AD}_f}Q \text{ for some permutation matrix } Q\\
         & \Longrightarrow \mathcal{C}_{\mathbb{AD}_g} \text{ and }\mathcal{C}_{\mathbb{AD}_f} \text{ are equivalent}.
    \end{split}
\end{equation*}	

(iii) $\Longrightarrow$ (i): By Lemma \ref{lem4.3},  a generator matrix of $\mathcal{C}_{\mathbb{AD}_f}$ is $G_{f}$. Then $\mathcal{C}_{\mathbb{AD}_f}$ is the dual of the code $\mathcal C_1(f)$ defined in \cite{EP}. The result follows from the fact that two codes are equivalent if and only if their duals are equivalent, together with \cite[Theorem 9]{EP}.

This completes the proof.
\end{proof}

\begin{remark}{\rm  \label{remarkff22} Theorem \ref{thm4.8} does not hold for plateaued functions with linear structure. Indeed,  if $f$ has nonzero linear structure, then $\mathcal{C}_{\mathbb{AD}_{f}}$ is a subcode of $\mathcal{C}_{f}$, where $\mathcal{C}_{f}=\{c(t, u, s)=(t, u, s)G_{f}: (t, u, s) \in \mathbb{F}_{2} \times \mathbb{F}_{2}^{m} \times \mathbb{F}_{2}\}$ and $G_f$ was defined in Lemma \ref{lem4.3}. In general, the existence of equivalent subcodes does not suffice to guarantee the equivalence of the two codes.

}
\end{remark}




 Prior studies \cite{DMT, DT0, DT, DT1, T,TD, WTD} rely on using the (generalized) Assmus-Mattson theorem or investigating the transitivity of the automorphism group of a linear code. Here, we give a direct proof.

\begin{proposition}\label{propro1} {\rm
    Let $f$ be a bent function in $\mathcal{BF}_m$. The following  statements hold:
    \begin{itemize}
        \item [(i)] The design held by the minimum weight codewords of $\mathcal{C}_{\mathbb{AD}_f}$ is $\mathbb{AD}_{f^{*}}$.
        \item[(ii)] The design held by  the codewords with Hamming weight $2^{m-1}+2^{m/2-1}$ of $\mathcal{C}_{\mathbb{AD}_f}$ is the complementary design of $\mathbb{AD}_{f^{*}}$ in (i). 
    \end{itemize} }
\end{proposition}

\begin{proof}

Observe that the design held by  codewords of weight $j\neq 0$ in $\C$ generated by $M(\mathbb{D})$ for any 2-design in terms of $(f_b(p))_{p\in P, b\in B}=uG$; the codewords of weight $j$ relies on $u$. By \cite[Theorem 12]{DMT}, the weight distribution of $\mathcal{C}_{\mathbb{AD}_f}$ is given by $$1+2^{m}z^{2^{m-1}-2^{\frac{m}{2}-1}}+2^{m} z^{2^{m-1}+2^{\frac{m}{2}-1}}+(2^{m+1}-2)z^{2^{m-1}}+z^{2^{m}}.$$

(i). The minimum weight codewords of $\mathcal{C}_{\mathbb{AD}_f}$ are $(f^{*}(p)+b \cdot p+f(b))_{p \in \mathbb{F}_{2}^{m}}$, where $b \in \mathbb{F}_{2}^{m}$. Since the dual $f^{*}$ satisfies $(f^{*})^{*}(x)=f(x)$, we have $f^{*}_{b}(p)=f^{*}(p)+b \cdot p+f(b)$. Thus, the support of the codeword $(f^{*}(p)+b \cdot p+f(b))_{p \in \mathbb{F}_{2}^{m}}$ is $\{p \in \mathbb{F}_{2}^{m}: f^{*}_{b}(p)=1\}$, which shows the design of the minimum weight codewords of $\mathcal{C}_{\mathbb{AD}_f}$ is $\mathbb{AD}_{f^{*}}$.

(ii). The codewords with weight $2^{m-1}+2^{m/2-1}$ of $\mathcal{C}_{\mathbb{AD}_f}$ are $(f^{*}(p)+b \cdot p+f(b)+1)_{p \in \mathbb{F}_{2}^{m}}$, where $b \in \mathbb{F}_{2}^{m}$. Thus, the support of the codeword $(f^{*}(p)+b \cdot p+f(b)+1)_{p \in \mathbb{F}_{2}^{m}}$ is $\{p \in \mathbb{F}_{2}^{m}: f^{*}_{b}(p)=0\}$, which shows the design of the codewords with weight $2^{m-1}+2^{m/2-1}$ of $\mathcal{C}_{\mathbb{AD}_f}$ is the complement of $\mathbb{AD}_{f^{*}}$.
\end{proof}

\begin{corollary} \label{cor4} {\rm (\cite{DS,EP, DMT})
Let $f$ and $g$ be bent functions in $\mathcal{BF}_m$. 
     The following  statements are equivalent: 
 \begin{itemize}
     \item[(i)] $f$ and $g$ are equivalent.
     
     \item[(ii)] $\mathbb{AD}_f$ and $\mathbb{AD}_g$ defined in  Lemma \ref{thm 2.23} are isomorphic.
     
     \item[(iii)] $\mathcal{C}_{\mathbb{AD}_f}$ and $\mathcal{C}_{\mathbb{AD}_g}$ are equivalent.
     
      \item[(iv)] The designs held by the minimum weight codewords of $\mathcal{C}_{\mathbb{AD}_f}$ and $\mathcal{C}_{\mathbb{AD}_g}$ are isomorphic.

      \item[(v)]  The designs held by the codewords with Hamming weight  $2^{m-1}+2^{m/2-1}$ of $\mathcal{C}_{\mathbb{AD}_f}$ and $\mathcal{C}_{\mathbb{AD}_g}$ are isomorphic. 
      
 \end{itemize}

}

\end{corollary}

We remark that the number of non-isomorphic TSDP designs with parameters $(\dag)$ grows exponentially when $m$ grows to infinity \cite{K1}. Further when $f$ is bent in $\mathcal{BF}_m$, $f, \mathbb{AD}_f, \mathcal{C}_{\mathbb{AD}_f}$, and the designs held by the minimum weight codewords of $\mathbb{AD}_f$ grow exponentially when $m$ grows to infinity \cite{DS,DMT}. 

\begin{remark}\label{rmk5}{\rm 
Note that item (iv) in Corollary \ref{cor4} cannot be applied to the case of $r$-plateaued functions, where $r>0$. We write $W_{f}(x)=(-1)^{g(x)}2^{\frac{m+r}{2}}1_{S_{f}}(x)$ for $g \in \mathcal{BF}_{m}$. By  \cite[Theorem 5]{WH}, the minimum weight codewords of $\mathcal{C}_{\mathbb{AD}_f}$ are $(g(b)+b \cdot p+f(p))_{p \in \mathbb{F}_{2}^{m}}$, where $b \in S_{f}$. Let $f'_{b}(p)=g(b)+b \cdot p+f(p)$, where $b \in S_{f}$ and $p \in \mathbb{F}_{2}^{m}$. Thus the support of the minimum weight codeword $(g(b)+b \cdot p+f(p))_{p \in \mathbb{F}_{2}^{m}}$ is $\{p \in \mathbb{F}_{2}^{m}: f'_{b}(p)=1\}$. By Corollary \ref{cor3.2}, if the minimum weight codewords support a design, then for any distinct $p, q \in \mathbb{F}_{2}^{m}$, $\sum_{b \in S_{f}}(-1)^{f'_{b}(p)+f'_{b}(q)}$ should be a constant. However, $\sum_{b \in S_{f}}(-1)^{f'_{b}(p)+f'_{b}(q)}=(-1)^{f(p)+f(q)}\sum_{b \in S_{f}}(-1)^{(p+q) \cdot b}$, which is not a constant in general.
}
\end{remark}

\subsection{On the design $\mathbb{D}_f$ derived from an $r$-plateaued function}


In this subsection, we will present a new construction of a linear code $\widetilde{\mathcal{C}}_{D_f}$ by using an $r$-plateaued function and then consider some equivalent problems. 

For $f\in \mathcal{BF}_m$, denote by $D_f=f^{-1}(1)$ an ordered subset of $\mathbb{F}^m_2$. We define a linear code $\widetilde{\mathcal{C}}_{D_f}$ of length $|D_f|$ by
\begin{align}\label{code}
    \widetilde{\mathcal{C}}_{D_f}=\{(x\cdot y)_{y\in D_f}+s{\bf 1}:x\in\mathbb{F}^m_2, s\in\mathbb{F}_2\},
\end{align}
where the all-one vector is denoted as $\bf{1}$.
If $f$ is $r$-plateaued, then the  weight distribution of $\widetilde{\mathcal{C}}_{D_f}$ could be determined by \cite[Theorem 1]{DC}, \cite{DT0}. Notice that
\[G_f'=\begin{pmatrix}
\cdots & x^T & \cdots \\
\cdots & f(x) & \cdots
\end{pmatrix}_{x\in D_f} 
\]
is a generator matrix of $\widetilde{\mathcal{C}}_{D_f}$.
 In the following, we characterize the equivalence of $\widetilde{\mathcal{C}}_{D_f}$.


\begin{theorem}{\rm \label{thmmm}
Let $r$ be an integer with $0 \leq r <m$. Let $f$ and $g$ be $r$-plateaued functions in $\mathcal{BF}_{m}$ with no nonzero linear structure. Then $\widetilde{\mathcal{C}}_{D_f}$ and $\widetilde{\mathcal{C}}_{D_g}$ defined in Eq.  (\ref{code}) are equivalent if and only if $f$ and $g$ are affine equivalent. In particular, if $\widetilde{\mathcal{C}}_{D_f}$ and $\widetilde{\mathcal{C}}_{D_g}$ are equivalent, then  $\mathcal{C}_{\mathbb{AD}_f}$ and $\mathcal{C}_{\mathbb{AD}_g}$ are equivalent. 
}
\end{theorem}

\begin{proof}
($\Longleftarrow$):
If $f$ and $g$ are affine equivalent, then there is an affine permutation $\pi$ of $\mathbb{F}_{2}^{m}$ such that $f(x)=g(\pi x)$. Thus $x \in D_{f}$ if and only if $x \in \pi^{-1}D_{g}$. For any $a \in \mathbb{F}_{2}^{m}$ and $b \in \mathbb{F}_{2}$, 
\begin{equation*}
(a \cdot x+b)_{x \in D_{f}}=(a \cdot x+b)_{x \in \pi^{-1}D_{g}}=(a \cdot \pi^{-1}x+b)_{x \in D_{g}},
\end{equation*}
which implies that $\widetilde{\mathcal{C}}_{D_f}$ and $\widetilde{\mathcal{C}}_{D_g}$ are equivalent.

($\Longrightarrow$): Let $\widetilde{\mathcal{C}}_{D_f}$ and $\widetilde{\mathcal{C}}_{D_g}$ be equivalent. If $(m, r)=(2, 0)$, then 
by $|D_{f}|=|D_{g}|$ and all bent functions are with the form $x_{1}x_{2}+c_{1}x_{1}+c_{2}x_{2}+c_{0}$ (where $c_{0}, c_{1}, c_{2} \in \mathbb{F}_{2}$), we have $f(x_{1}, x_{2}), g(x_{1}, x_{2}) \in \{x_{1}x_{2}, x_{1}(x_{2}+1), (x_{1}+1)x_{2}, (x_{1}+1)(x_{2}+1)\}$, or $f(x_{1}, x_{2}), g(x_{1}, x_{2}) \in \{x_{1}x_{2}+1, x_{1}(x_{2}+1)+1, (x_{1}+1)x_{2}+1, (x_{1}+1)(x_{2}+1)+1\}$. Hence, $f$ and $g$ are affine equivalent. In the following, we consider the case $(m, r) \neq (2, 0)$. Since $\widetilde{\mathcal{C}}_{D_f}$ and $\widetilde{\mathcal{C}}_{D_g}$ are equivalent, then there is a bijection $\pi: D_{f} \rightarrow D_{g}$ such that for any $(a, b) \in \mathbb{F}_{2}^{m} \times \mathbb{F}_{2}$, there is a unique $(\mathcal{T}(a, b), \mathcal{U}(a, b)) \in \mathbb{F}_{2}^{m} \times \mathbb{F}_{2}$ satisfying 
$
    a \cdot x+b=\mathcal{T}(a, b) \cdot \pi x+\mathcal{U}(a, b) 
    \text{ for all } x \in D_{f}.
$
Simply denote $\mathcal{T}(a, 0)=\mathcal{T}(a)$ and $\mathcal{U}(a, 0)=\mathcal{U}(a)$. Then
\begin{equation}\label{eq14}
a \cdot x=\mathcal{T}(a) \cdot \pi x+\mathcal{U}(a) \text{ for all } a \in \mathbb{F}_{2}^{m} \text{ and } x \in D_{f}.
\end{equation}
By Eq.  \eqref{eq14}, we have
\begin{equation}\label{eq15}
\mathcal{T}(a) \cdot y=a \cdot \pi^{-1}y+\mathcal{U}(a) \text{ for all } a \in \mathbb{F}_{2}^{m} \text{ and } y \in D_{g}.
\end{equation}
Denote $D_{g}=y_{0}+E_{g}$, where $y_{0} \in D_{g}$. Note that ${\bf{0}} \in E_{g}$. By Eq. \eqref{eq15}, for any $a \in \mathbb{F}_{2}^{m}$, we have
\begin{equation} \label{eq16}
    \mathcal{T}(a) \cdot y_{0}=a \cdot \pi^{-1}(y_{0})+\mathcal{U}(a),
\end{equation}
\begin{equation}\label{eq17}
\mathcal{T}(a) \cdot (y_{0}+y)=a \cdot \pi^{-1}(y_{0}+y)+\mathcal{U}(a) \text{ for all } y \in E_{g}.
\end{equation}
Adding Eqs. \eqref{eq16} and \eqref{eq17}, for any $a \in \mathbb{F}_{2}^{m}$,
\begin{equation}\label{eq18}
    \mathcal{T}(a) \cdot y=a \cdot (\pi^{-1}(y_{0})+\pi^{-1}(y_{0}+y)) \text{ for all } y \in E_{g}.
\end{equation}
By Eq.  \eqref{eq18}, for any $a, a' \in \mathbb{F}_{2}^{m}$ and $y \in E_{g}$, we obtain
\begin{equation}\label{eq19}
\begin{split}
&(\mathcal{T}(a+a')+\mathcal{T}(a)+\mathcal{T}(a')) \cdot y\\
&=\mathcal{T}(a+a') \cdot y+\mathcal{T}(a) \cdot y+\mathcal{T}(a') \cdot y\\
& =(a+a') \cdot  (\pi^{-1}(y_{0})+\pi^{-1}(y_{0}+y))\\
&+a \cdot (\pi^{-1}(y_{0})+\pi^{-1}(y_{0}+y))\\
& \ \ \ +a' \cdot  (\pi^{-1}(y_{0})+\pi^{-1}(y_{0}+y))=0.
\end{split}
\end{equation}
We claim that for any $r$-plateaued function $g \in \mathcal{BF}_{m}$ with no nonzero structure (where $(m, r) \neq (2, 0)$), $E_{g}$ contains a basis of $\mathbb{F}_{2}^{m}$. For any $a \in \mathbb{F}_{2}^{m}$, we have
\begin{equation}\label{eq20}
    \begin{split}
        W_{g}(a)&=\sum_{y \in \mathbb{F}_{2}^{m}}(-1)^{g(y)+a \cdot y}=-\sum_{y \in D_{g}}(-1)^{a \cdot y}+\sum_{y \in \mathbb{F}_{2}^{m} \backslash D_{g}}(-1)^{a \cdot y}\\
        &=\sum_{y \in \mathbb{F}_{2}^{m}}(-1)^{a \cdot y}-2\sum_{y \in D_{g}}(-1)^{a \cdot y}\\
        &=2^m \delta_{{\bf{0}},a}-2(-1)^{a \cdot y_{0}}\sum_{y \in E_{g}}(-1)^{a \cdot y}.
    \end{split}
\end{equation}
Denote $E_{g}^{\bot}=\{a \in \mathbb{F}_{2}^{m}: a \cdot y=0 \text{ for all } y \in E_{g}\}$. As easily seen, $E_{g}^{\bot}$ is a linear subspace and $E_{g}^{\bot}=\LI E_{g}\RI^{\bot}$. Assume that $E_{g}$ does not contain any basis of $\mathbb{F}_{2}^{m}$. Then there is nonzero $a_{0} \in E_{g}^{\bot}$, and by Eq. \eqref{eq20}, 
\begin{equation}\label{eq21}
|W_{g}(a_{0})|=2|E_{g}|.
\end{equation}
Since $|E_{g}| \in \{2^{m-1}, 2^{m-1}+2^{\frac{m+r}{2}-1}, 2^{m-1}-2^{\frac{m+r}{2}-1}\}$ and $|W_{g}(a_{0})| \in \{0, 2^{\frac{m+r}{2}}\}$, then Eq. \eqref{eq21} implies that $r=m-2$ and $|E_{g}|=2^{m-2}$. Note that when $|E_{g}|=2^{m-2}$, then $g$ is unbalanced (that is, ${\bf{0}} \in S_{g}$), and $|E_{g}^{\bot}| \in \{2, 4\}$ as the dimension of $\LI E_{g}\RI$ is not less than $m-2$. If $r=m-2$ and $|E_{g}|=2^{m-2}$, then for any $a \in S_{g} \backslash \{{\bf{0}}\}$, by Eq. \eqref{eq20}, we have $ \pm 2^{m-1}=2(-1)^{a \cdot y_{0}+1}\sum_{y \in E_{g}}(-1)^{a \cdot y}$, which implies that $\sum_{y \in E_{g}}(-1)^{a \cdot y}= \pm 2^{m-2}=\pm|E_{g}|$. As ${\bf{0}} \in E_{g}$, then for any $a \in S_{g} \backslash \{{\bf{0}}\}$, we have $a \cdot y=0$ for all $y \in E_{g}$, that is, $a \in E_{g}^{\bot} \backslash \{{\bf{0}}\}$. Hence $S_{g} \subseteq E_{g}^{\bot}$. By $|S_{g}|=2^{m-r}=4$ and $|E_{g}^{\bot}| \in \{2, 4\}$, we have $S_{g}=E_{g}^{\bot}$ is a $2$-dimensional subspace of $\mathbb{F}_{2}^{m}$, which contradicts that $g$ has no nonzero linear structure and $(m, r) \neq (2, 0)$. By the above arguments, $E_{g}$ contains a basis of $\mathbb{F}_{2}^{m}$. 

Since $E_{g}$ contains a basis of $\mathbb{F}_{2}^{m}$, then by Eq. \eqref{eq19}, we obtain $\mathcal{T}(a+a')+\mathcal{T}(a)+\mathcal{T}(a')=0$ for any $a, a' \in \mathbb{F}_{2}^{m}$. Therefore $\mathcal{T}(a)$ is a linear function from $\mathbb{F}_{2}^{m}$ to $\mathbb{F}_{2}^{m}$. For $a, a' \in \mathbb{F}_{2}^{m}$, if $\mathcal{T}(a)=\mathcal{T}(a')$, then $\mathcal{U}(a) \neq \mathcal{U}(a')$, and by Eq. \eqref{eq14}, $(a+a') \cdot x=1$ for all $x \in D_{f}$. Denote $D_{f}=x_{0}+E_{f}$, where $x_{0} \in D_{f}$. Then $(a+a') \cdot x=1+(a+a') \cdot x_{0}$ for all $x \in E_{f}$. Since ${\bf{0}} \in E_{f}$, we obtain $1+(a+a') \cdot x_{0}=(a+a') \cdot {\bf{0}}=0$, and then $(a+a') \cdot x=0$ for all $x \in E_{f}$. Recall that we have proved that $E_{f}$ contains a basis of $\mathbb{F}_{2}^{m}$, thus $a=a'$, and $\mathcal{T}(a)$ is a linear permutation of $\mathbb{F}_{2}^{m}$. By Eq. \eqref{eq14}, for any $a \in \mathbb{F}_{2}^{m}$,
\begin{equation}\label{eq22}
    a \cdot x_{0}=\mathcal{T}(a) \cdot \pi x_{0}+\mathcal{U}(a),
\end{equation}
\begin{equation}\label{eq23}
    a \cdot (x_{0}+x)=\mathcal{T}(a) \cdot \pi(x_{0}+x)+\mathcal{U}(a) \text{ for all } x \in E_{f}.
\end{equation}
Let $\mathcal{T}^{*}(a)$ be the adjoint permutation of $\mathcal{T}(a)$, that is, for any $a, z \in \mathbb{F}_{2}^{m}$, $\mathcal{T}(a) \cdot z=a \cdot \mathcal{T}^{*}(z)$. Note that $\mathcal{T}^{*}$ is linear. Adding Eqs. \eqref{eq22} and \eqref{eq23}, for any $a \in \mathbb{F}_{2}^{m}$ and $x \in E_{f}$, we have
\begin{eqnarray*}
    &a \cdot x =\mathcal{T}(a) \cdot (\pi x_{0}+\pi(x_{0}+x))\\
    &=a \cdot \mathcal{T}^{*}(\pi x_{0}+\pi(x_{0}+x)) \text{ for all } a\in\F^m_2,
\end{eqnarray*}
which implies that 
\begin{equation}\label{eq24}
    x=\mathcal{T}^{*}(\pi x_{0}+\pi(x_{0}+x)) \text{ for all } x \in E_{f}.
\end{equation}
By Eq. \eqref{eq24}, 
$
    \pi x=(\mathcal{T}^{*})^{-1}(x_{0}+x)+\pi x_{0} \text{ for all } x \in D_{f}.
$
Extending $\pi: \mathbb{F}_{2}^{m} \rightarrow \mathbb{F}_{2}^{m}$ by
$\pi x=(\mathcal{T}^{*})^{-1}(x_{0}+x)+\pi x_{0} \text{ for any } x \in \mathbb{F}_{2}^{m}$. Since $(\mathcal{T}^{*})^{-1}$ is a linear permutation of $\mathbb{F}_{2}^{m}$, then it is easy to verify that $\pi$ is an affine permutation of $\mathbb{F}_{2}^{m}$. Since $\pi$ is an affine permutation with $\pi D_{f}=D_{g}$, then $f(x)=g(\pi x) \text{ for all } x \in \mathbb{F}_{2}^{m}$, that is, $f$ and $g$ are affine equivalent.
\end{proof}

 
 Note that  the converse statement  \lq\lq if $\widetilde{\mathcal{C}}_{D_f}$ and $\widetilde{\mathcal{C}}_{D_g}$ are equivalent, then  $\mathcal{C}_{\mathbb{AD}_f}$ and $\mathcal{C}_{\mathbb{AD}_g}$ are equivalent\rq\rq\; in Theorem \ref{thmmm} is not true, see the following counterexample.

\begin{example}{\rm
Let $f, g$ be two Boolean functions in $\mathcal{BF}_{5}$ be given by $f(x_{1}, x_{2}, x_{3}, x_{4}, x_{5})=x_{1}x_{3}+x_{2}x_{4}+x_{1}x_{2}x_{5}$,  and $g(x_{1}, x_{2}, x_{3}, x_{4}, x_{5})=f(x_{1}, x_{2}, x_{3}, x_{4}, x_{5})+x_{1}+x_{2}+x_{3}+x_{4}$. Then $f, g$ are $1$-plateaued functions with no nonzero linear structure by  \cite[Example 4.1]{HP}. Since $f$ and $g$ are equivalent, then $\mathcal{C}_{\mathbb{AD}_f}$ and $\mathcal{C}_{\mathbb{AD}_g}$ are equivalent by Theorem \ref{thm4.8}. However, we have $|D_{f}|=12$ and $|D_{g}|=16$ and thus $\widetilde{\mathcal{C}}_{D_f}$ and $\widetilde{\mathcal{C}}_{D_g}$ are not equivalent.
}
\end{example}

{ To obtain the equivalence problem analogous to Corollary \ref{cor4}, we construct a new 2-design. Only in the following theorem, we use $\mathbb{F}^m_2\setminus\{\bf{0}\}$ instead of $\mathbb{F}^{m*}_2$.}
\begin{theorem}{\rm  \label{theorembent}

    Let $f$ be a bent function in $\mathcal{BF}_m$. For each $b\in \mathbb{F}^m_2$, let $f_b(x)=b\cdot x+f^*(b)$ for $x\in \mathbb{F}^m_2$. Then  a pair $(D_f,\{f_b\in\mathcal{BF}_m:b\in \mathbb{F}^{m}_2\setminus\{\bf{0}\}\})$ induces a $2$-design $\mathbb{D}_f=(D_f,\{B^{f_b}:b\in\mathbb{F}^{m}_2\setminus\{\bf{0}\}\})$.

    \begin{itemize}
		\item [(i)] The parameters of the $2$-design $\mathbb{D}_f$ are 
     $2$-$(2^{m-1}-(-1)^{f^{*}(\mathbf{0})}2^{\frac{m}{2}-1}, 2^{m-2}-2^{\frac{m}{2}-2}((-1)^{f^{*}(\mathbf{0})}-1), 2^{m-2}{+2^{\frac{m}{2}-1}}+\frac{1}{2}((-1)^{f^{*}(\mathbf{0})}-1))$.

    \item [(ii)] The block intersection numbers of two distinct vectors $a, b \in {\mathbb{F}_{2}^{m}}\setminus\{\mathbf{0}\}$ are 
$$2^{m-3}+2^{\frac{m}{2}-2}-2^{\frac{m}{2}-3}((-1)^{f^{*}(\mathbf{0})}+(-1)^{f^{*}(a)+f^{*}(b)+f^{*}(a+b)}),$$
so that $\mathbb{D}_f$ is quasi-symmetric.

\item [(iii)] The  2-rank of $\mathbb{D}_f$ is $m+1$.

\end{itemize}

    }
\end{theorem}


\begin{proof}
(i). We will check the three conditions of Corollary \ref{cor3.2}. Firstly, for $b\in\mathbb{F}^{m}_2\setminus\{{\bf{0}}\},$ we have
 \begin{eqnarray*}
&W_{D_f,f_b}(\mathbf{0})=\sum_{p\in D_f}(-1)^{b\cdot p+f^*(b)}\\
&=(-1)^{f^*(b)}(2^{m-1}\delta_{{\bf{0}},b}-2^{\frac{m}{2}-1}(-1)^{f^*(b)})=-2^{\frac{m}{2}-1}.
 \end{eqnarray*}
Next,  by symmetry that for $p \in D_{f}$, we have
 \begin{eqnarray*}
&\sum_{b\in\mathbb{F}^{m}_2\setminus\{{\bf{0}}\}}(-1)^{f_b(p)}=-(-1)^{f^{*}(\mathbf{0})}+W_{\mathbb{F}^m_2,f^*}(p)\\
&=-(-1)^{f^{*}(\mathbf{0})}-2^{\frac{m}{2}}.
 \end{eqnarray*}
Finally, for two distinct points $p,q\in D_f$,
\[
\sum_{b\in \mathbb{F}^{m}_2\setminus\{{\bf{0}}\}}(-1)^{f_b(p)+f_b(q)}=\sum_{b\in \mathbb{F}^{m}_2\setminus\{{\bf{0}}\}}(-1)^{(p+q)\cdot b}=-1.
\] 
The parameters are computed by those three values of summations using   Corollary \ref{cor3.2}. 

(ii). The point-block intersection numbers are 
  \begin{equation*}
    \begin{split}
   & \frac{1}{4}(|D_{f}|-\sum_{p \in D_{f}}(-1)^{f_{a}(p)}-\sum_{p \in D_{f}}(-1)^{f_{b}(p)}+\sum_{p \in D_{f}}(-1)^{f_{a}(p)+f_{b}(p)})\\
   &=\frac{1}{4}(2^{m-1}-(-1)^{f^{*}(\mathbf{0})}2^{\frac{m}{2}-1}+2^{\frac{m}{2}-1}+2^{\frac{m}{2}-1}\\
   &+(-1)^{f^{*}(a)+f^{*}(b)}(2^{m-1}\delta_{\mathbf{0}, a+b}-2^{\frac{m}{2}-1}(-1)^{f^{*}(a+b)}))\\
   &=2^{m-3}+2^{\frac{m}{2}-2}-2^{\frac{m}{2}-3}((-1)^{f^{*}(\mathbf{0})}+(-1)^{f^{*}(a)+f^{*}(b)+f^{*}(a+b)}).
    \end{split}
\end{equation*}  

(iii).  
Notice that for any $p \in D_{f}$, $(f_{b}(p))_{b \in {\mathbb{F}_{2}^{m}}\setminus\{{\bf{0}}\}}=(p, f(p))G_{f^{*}}^{''}$, where $G_{f^{*}}^{''}=\begin{pmatrix}
\cdots & b^T & \cdots \\
\cdots & f^{*}(b) & \cdots
\end{pmatrix}_{b\in {\mathbb{F}_{2}^{m}}\setminus\{{\bf{0}}\}}$. By \cite[Theorem 18.3.4]{M}, the rank of $G_{f^{*}}^{''}$ is $m+1$. By Lemma \ref{lem3.6}, the binary linear code $\mathcal{C}_{\mathbb{D}_{f}}$ is $\LI\{(p, f(p))G_{f^{*}}^{''}: p \in D_{f}\}\RI=\{(t, u)G_{f^{*}}^{''}: (t, u) \in \LI\{(p, f(p)): p \in D_{f}\}\RI\}$. 
Since the dimension of $\widetilde{\mathcal{C}}_{D_f}$ is $m+1$, that is, $\LI\{(p, f(p)): p \in D_{f}\}\RI=\mathbb{F}_{2}^{m} \times \mathbb{F}_{2}$, then $G_{f^{*}}^{''}$ is a generator matrix of $\mathcal{C}_{\mathbb{D}_{f}}$, and the rank of $\mathbb{D}_{f}$ is $m+1$.
This completes the proof.
\end{proof}

\begin{remark}
    {\rm  The block graph of a quasi-symmetric design with intersection number $i,j~ (i<j)$, where two
blocks are adjacent if they intersect in $i$ points, is strongly regular  \cite{CD}.

    }
\end{remark}


\begin{proposition}\label{propro}{\rm 
    Let $f$ be a bent function in $\mathcal{BF}_m$. The following  statements hold:
    \begin{itemize}

     \item[(i)] The design {held} by the middle weight codewords of $\widetilde{\mathcal{C}}_{D_f}$ is $\mathbb{D}_f$ in Theorem \ref{theorembent}.
     
        \item [(ii)] The design  {held} by the minimum weight codewords of $\widetilde{\mathcal{C}}_{D_f}$ is the complement of $\mathbb{D}_f$.

    \end{itemize}
    }
\end{proposition}
\begin{proof}

 (i). By \cite[Theorem 1]{DC}, the middle weight codewords of $\widetilde{\mathcal{C}}_{D_f}$ are $(b \cdot p+f^{*}(b))_{p \in D_{f}}$, where $b \in {\mathbb{F}_{2}^{m}}^{*}$. Recall that $f_{b}(x)=b \cdot x+f^{*}(b)$. Then the support of the middle weight codeword $(b \cdot p+f^{*}(b))_{p \in D_{f}}$ is $\{p \in D_{f}: f_{b}(p)=1\}$, where $b \in {\mathbb{F}_{2}^{m}}^{*}$. Thus the design of the middle weight codewords of $\widetilde{\mathcal{C}}_{D_f}$ is  $\mathbb{D}_f$.

  (ii). By \cite[Theorem 1]{DC}, the minimum weight codewords of $\widetilde{\mathcal{C}}_{D_f}$ are $(b \cdot p+f^{*}(b)+1)_{p \in D_{f}}$, where $b \in {\mathbb{F}_{2}^{m}}^{*}$. Then the support of the minimum weight codeword $(b \cdot p+f^{*}(b)+1)_{p \in D_{f}}$ is $\{p \in D_{f}: f_{b}(p)=0\}$, where $b \in {\mathbb{F}_{2}^{m}}^{*}$. Thus, the design of the minimum weight codewords of $\widetilde{\mathcal{C}}_{D_f}$ is the complement of $\mathbb{D}_f$.
\end{proof}

The following theorem  completely solves two open problems mentioned in introduction.

 \begin{theorem} \label{cor5} {\rm
Let $f$ and $g$ be bent functions in $\mathcal{BF}_m$. 
 The following statements are equivalent:
\begin{itemize}
     \item[(i)] $f$ and $g$ are affine equivalent.

   \item[(ii)] $\mathbb{D}_f$ and $\mathbb{D}_g$ defined in Theorem \ref{theorembent} are isomorphic.

     \item[(iii)] $\widetilde{\mathcal{C}}_{D_f}$ and $\widetilde{\mathcal{C}}_{D_g}$ defined in Eq. (\ref{code}) are equivalent.

      \item[(iv)]  The designs held by  the minimum weight codewords of $\widetilde{\mathcal{C}}_{D_f}$ and $\widetilde{\mathcal{C}}_{D_g}$ are isomorphic. 
    
     \item[(v)]  The designs held by  the middle weight codewords of $\widetilde{\mathcal{C}}_{D_f}$ and $\widetilde{\mathcal{C}}_{D_g}$ are isomorphic. 
    
 \end{itemize}

 }
\end{theorem}


\begin{proof}

(i)  $\Longleftrightarrow$  (iii): The result follows from  Theorem \ref{thmmm}. 

(iv)  $\Longleftrightarrow$  (v): The result follows from  Proposition \ref{propro}.

 (ii)  $\Longleftrightarrow$  (iv):  The result follows from Proposition \ref{propro}.

(i) $\Longrightarrow$ (ii): Since $f$ and $g$ are affine equivalent, then there is an affine permutation $\pi$ of $\mathbb{F}_{2}^{m}$ such that $g(x)=f(\pi x)$. Observe that $\pi(D_{g})=D_{f}$. Denote $\pi x=xA+c$ for $A \in \operatorname{GL}_{m}(\mathbb{F}_{2})$ and $c \in \mathbb{F}_{2}^{m}$. For any $b \in \mathbb{F}_{2}^{m}$, we have
\begin{eqnarray*}
&W_{g}(b)=\sum_{x \in \mathbb{F}_{2}^{m}}(-1)^{g(x)+b \cdot x}\\
&=\sum_{x \in \mathbb{F}_{2}^{m}}(-1)^{f(xA+c)+b \cdot x}\\
&=\sum_{x \in \mathbb{F}_{2}^{m}}(-1)^{f(x)+b \cdot (x+c)A^{-1}}\\
&=\sum_{x \in \mathbb{F}_{2}^{m}}(-1)^{f(x)+b(A^{-1})^{T} \cdot x+b \cdot cA^{-1}}\\
&=(-1)^{b \cdot cA^{-1}}W_{f}(b(A^{-1})^{T})=2^{\frac{m}{2}}(-1)^{f^{*}(b(A^{-1})^{T})+b \cdot cA^{-1}}.
\end{eqnarray*}
Thus, 
\begin{equation}\label{Eq34}
    g^{*}(b)=f^{*}(b(A^{-1})^{T})+b \cdot cA^{-1}.
\end{equation}
Let $\sigma x=x(A^{-1})^{T}$. Note that $\sigma{\mathbb{F}_{2}^{m}}^{*}={\mathbb{F}_{2}^{m}}^{*}$. Then by Eq.  \eqref{Eq34}, for any $b \in {\mathbb{F}_{2}^{m}}^{*}$ and $x \in D_{g}$, 
\begin{equation*}
    b \cdot x+g^{*}(b)=b \cdot (x+cA^{-1})+f^{*}(b(A^{-1})^{T})=\sigma b \cdot \pi x+f^{*}(\sigma b),
\end{equation*}
which shows that $\mathbb{D}_{f}$ and $\mathbb{D}_{g}$ are isomorphic.

(ii) $\Longrightarrow$ (i):  
By the proof of Theorem \ref{thmmm}, $f$ and $g$ are affine equivalent when $m=2$. In the following, let $m \geq 4$. Without loss of generality, let ${\bf{0}} \in D_{g}$. Since  $\mathbb{D}_{f}$ and $\mathbb{D}_{g}$ are isomorphic, then there is a permutation $\sigma$ of ${\mathbb{F}_{2}^{m}}^{*}$ and a bijection $\pi$ from $D_{g}$ to $D_{f}$, such that for any $b \in {\mathbb{F}_{2}^{m}}^{*}$ and $x \in D_{g}$, $
b \cdot x+g^{*}(b)=\sigma b \cdot \pi x+f^{*}(\sigma b).
$
Let us extend $\sigma$ by defining $\sigma{\bf{0}}={\bf{0}}$. Then $\sigma$ is a permutation of $\mathbb{F}_{2}^{m}$. By the parameters of $\mathbb{D}_{f}$ and $\mathbb{D}_{g}$, we have $g^{*}({\bf{0}})=f^{*}({\bf{0}})$. Since $\sigma{\bf{0}}={\bf{0}}$ and  $g^{*}( {\bf{0}})=f^{*}({\bf{0}})$, then for any $x \in D_{g}$, 
\begin{equation}\label{Eq35}
    {\bf{0}} \cdot x+g^{*}({\bf{0}})=\sigma {\bf{0}} \cdot \pi x+f^{*}(\sigma {\bf{0}})
\end{equation}
Therefore for any $b \in \mathbb{F}_{2}^{m}$ and $x \in D_{g}$, we have
\begin{equation}\label{Eq36}
b \cdot x+g^{*}(b)=\sigma b \cdot \pi x+f^{*}(\sigma b).
\end{equation}
Plugging $x= {\bf{0}}$ into Eq. \eqref{Eq36}, for any $b \in \mathbb{F}_{2}^{m}$, we have
\begin{equation}\label{Eq37}
    g^{*}(b)=\sigma b \cdot \pi  {\bf{0}}+f^{*}(\sigma b).
\end{equation}
Adding Eqs.~\eqref{Eq36} and  \eqref{Eq37}, for any $b \in \mathbb{F}_{2}^{m}$ and $x \in D_{g}$, we have
\begin{equation}\label{Eq38}
b \cdot x=\sigma b \cdot (\pi x+\pi {\bf{0}}).
\end{equation}
For any $b, b' \in \mathbb{F}_{2}^{m}$ and $x \in D_{g}$, by Eq. \eqref{Eq38}, we have
\begin{eqnarray}\label{Eq39}
    &\sigma (b+b') \cdot (\pi x+ \pi {\bf{0}})=(b+b') \cdot x=b \cdot x+b' \cdot x\nonumber\\
    &= (\sigma b+\sigma b') \cdot (\pi x+ \pi {\bf{0}}).
\end{eqnarray}
Let $E_{f}=D_{f}+\pi {\bf{0}}$. By the proof of Theorem \ref{thmmm}, $E_{f}$ contains a basis of $\mathbb{F}_{2}^{m}$. Hence, Eq. \eqref{Eq39} implies that $\sigma(b+b')=\sigma b+\sigma b'$ for any $b, b' \in \mathbb{F}_{2}^{m}$, that is, $\sigma$ is a linear permutation of $\mathbb{F}_{2}^{m}$. Denote $\sigma x=xA$ for $A \in \operatorname{GL}_{m}(\mathbb{F}_{2})$ and $c \in \mathbb{F}_{2}^{m}$. By Eq. \eqref{Eq37}, for any $x \in \mathbb{F}_{2}^{m}$. we have
\begin{equation*}
\begin{split}
   &2^{\frac{m}{2}}(-1)^{g(x)}=W_{g^{*}}(x)=\sum_{b \in \mathbb{F}_{2}^{m}}(-1)^{f^{*}(bA)+bA \cdot \pi {\bf{0}}+b \cdot x}\\
   &=\sum_{b \in \mathbb{F}_{2}^{m}}(-1)^{f^{*}(b)+b \cdot (\pi {\bf{0}}+x(A^{-1})^{T})}\\
   &=W_{f^{*}}(\pi {\bf{0}}+x(A^{-1})^{T})=2^{\frac{m}{2}}(-1)^{f(\pi {\bf{0}}+x(A^{-1})^{T}))}.
\end{split}
\end{equation*}
Thus $g$ and $f$ are affine equivalent with $g(x)=f(\pi {\bf{0}}+x(A^{-1})^{T})).$

This completes this proof.
\end{proof}

\begin{corollary}{\rm The number of non affine-equivalent bent functions in $\mathcal{BF}_m$ grow exponentially when $m$ grows to infinity. Further, when $f$ is bent in $\mathcal{BF}_m$, $\mathbb{D}_f, \widetilde{\mathcal{C}}_{\mathbb{D}_f}$, and the designs held by the minimum (resp. middle) weight codewords of $\widetilde{\mathcal{C}}_{\mathbb{D}_f}$ grow exponentially when $m$ grows to infinity.

}
    
\end{corollary}

\begin{proof}  Since the number of non-isomorphic TSDP designs with parameters $(\dag)$ grows exponentially when $m$ grows to  infinity (see \cite{K1}), it is only to show that if there are two non-isomorphic designs $\mathbb{D}_1, \mathbb{D}_2$, then there exist two bent functions $f, g$ such that they are not affine equivalent.

Suppose that $f$ and $g$ are affine equivalent. By Theorem \ref{cor5},  $\widetilde{\mathcal{C}}_{D_f}$ and $\widetilde{\mathcal{C}}_{D_g}$ are equivalent. By Theorem \ref{thmmm}, $\mathcal{C}_{\mathbb{AD}_f}$ and $\mathcal{C}_{\mathbb{AD}_g}$ are equivalent. By Theorem 2, the two addition designs {$\mathbb{AD}_f$} and {$\mathbb{AD}_g$} are both TSDP.    This completes the proof by using Theorem \ref{cor5}.
\end{proof}


The following open problem   is quite natural.

\begin{problem}{\rm \label{openpro1}
 Can we extended those bent functions in Theorem \ref{theorembent} to $r$-plateaued functions? 

}
    
\end{problem}

\section{The automorphisms of  designs and linear codes induced by plateaued functions}\label{s6}

Bending \cite{B} computed the
automorphism groups of addition designs of bent functions. Later, Dempwolff and Neumann \cite{DN} computed the automorphism groups of addition designs of 1-plateaued functions which do not have a linear structure. In this section, we extend their results by computing the automorphism groups of addition designs of $r$-plateaued functions which do not have a linear structure. Our proof uses a completely different approach compared to \cite{DN}, and further we present the automorphism groups of the linear codes of addition designs.

\begin{theorem}\label{thm10}{\rm
 Let $f$ be an $r$-plateaued function in $\mathcal{BF}_m$ with no nonzero linear structure. We can write $W_{f}(x)=(-1)^{g(x)}2^{\frac{m+r}{2}}1_{S_{f}}(x)$ for some $g \in \mathcal{BF}_{m}$. Then  
 \begin{itemize}
     \item [(i)] $\Aut(\mathbb{AD}_f)$ is isomorphic to  $\Aut(f)$ given by
     $${(\sigma,\pi)\mapsto [A, \pi{\bf{0}}, \sigma{\bf{0}}A^{T}, \pi{\bf{0}}\cdot \sigma{\bf{0}}+g({\bf{0}})+g(\sigma{\bf{0}})]}.$$
      
     \item[(ii)] $\Aut(\mathcal{C}_{\mathbb{AD}_f})=\{(A,a)\in\operatorname{GL}_m(\mathbb{F}_2)\times\mathbb{F}^m_2:f(x)+f(xA+a) \text{ is an affine function in } \mathcal{BF}_m\}$.

    \item[(iii)]  
          $\Aut(\widetilde{\C}_{D_f})=\{\pi|_{D_{f}}: \pi \text{ is an affine permutation of } \mathbb{F}_{2}^{m} \text{ with } f(x)=f(\pi x)\}.$


    \item[(iv)] If $f$ is bent, then  $ {\Aut(\mathbb{D}_f)=\{(\pi |_{D_{f}}, (\bar{\pi}^{-1})^{*}|_{{\mathbb{F}_{2}^{m}}^{*}}): \pi \text{ is an affine permutation of } \mathbb{F}_{2}^{m} \text{ with } f(x)=f(\pi x)\}},  $
     where $\bar{\pi} x=\pi x+\pi {\bf{0}}$,  $(\bar{\pi}^{-1})^{*}$ is the adjoint permutation of $\bar{\pi}^{-1}$, and ${\mathbb{F}_{2}^{m}}^{*}=\mathbb F^m_2\setminus \{\bf{0}\}$.

 \end{itemize}
 
 }
\end{theorem}

\begin{proof}
(i). If $\textbf{0} \notin S_{f}$, then for $f'(x)=f(x)+x \cdot a$ with $a \in S_{f}$, 
we have $\textbf{0} \in S_{f'}$ since $W_{f'}(\textbf{0})=W_{f}(a)$. By Theorem \ref{thm4.8}, $\mathbb{AD}_{f}$ and $\mathbb{AD}_{f'}$ are isomorphic. Therefore, we may assume that $\textbf{0} \in S_{f}$. Let $(\sigma,\pi)\in\Aut(\mathbb{AD}_f)$. 
Then \begin{align}\label{eqn:14}
	f(x)+g(y)+x\cdot y =f(\pi x)+g(\sigma y)+\pi x\cdot \sigma y
	\end{align}
 for all $x\in\mathbb{F}_2^m$ and $y\in S_{f}$. By plugging $x=\textbf{0}$ into (\ref{eqn:14}), we obtain $f(\mathbf{0})+f(\pi\mathbf{0})+\pi\mathbf{0}\cdot \sigma y=g(y)+g(\sigma y)$ for all $y\in S_f$. By plugging $y=\mathbf{0}$ into (\ref{eqn:14}), we obtain $f(x)+f(\pi x)+\pi x\cdot\sigma \mathbf{0}=g(\mathbf{0})+g(\sigma \mathbf{0})$
for all $x\in\mathbb{F}_2^m$. By the above three equations, we obtain
$f(\mathbf{0})+f(\pi\mathbf{0})+x\cdot y+(\pi x+\pi{\mathbf{0}})\cdot(\sigma y+\sigma\mathbf{0})+\pi\mathbf{0}\cdot\sigma\mathbf{0}=g(\mathbf{0})+g(\sigma\mathbf{0})$
for all $x \in \mathbb{F}_{2}^{m} \text{ and } y \in S_{f}$, from which
we have $f(\mathbf{0})+f(\pi\mathbf{0})+\pi\mathbf{0}\cdot\sigma\mathbf{0}=g(\mathbf{0})+g(\sigma\mathbf{0})$ by putting $x=\mathbf{0}$. This implies that 
\begin{align}\label{eqn:17}
x\cdot y=\bar{\pi} x\cdot \bar{\sigma}y
\end{align}
for all $x\in\mathbb{F}_2^m$ and $y\in S_f$, where $\bar{\pi}x=\pi x+\pi\mathbf{0}, \bar{\sigma}y=\sigma y+\sigma\mathbf{0}$. 
Since $S_{f}$ contains a basis of $\mathbb{F}_{2}^{m}$, by \eqref{eqn:17}, it is easy to check that $\bar{\pi}x, \bar{\sigma}x\in\operatorname{GL}_m(\F_2)$. Then $\bar{\pi}x=xA$ and $\bar{\sigma}x=xB$ for some $A,B\in\operatorname{GL}_m(\F_2)$, and $x\cdot y=\bar{\pi}x\cdot\bar{\sigma}y=(xAB^T)\cdot y$ for all $x\in\F^m_2$ and $y \in S_{f}$. So $x(I_m+AB^T)\cdot y=0$, where $I_m$ is the identity matrix of size $m$. As $S_{f}$ contains a basis of $\mathbb{F}_{2}^{m}$, $AB^T=I_m$. It follows that $\pi x=xA+\pi{\bf{0}}$ and $\sigma x=x(A^T)^{-1}+\sigma{\bf{0}}$. Consequently, we have
\begin{align*}
     f(x)&=f(\pi x)+\pi x\cdot \sigma{\bf{0}}+g({\bf{0}})+g(\sigma{\bf{0}})\\
  &=f(xA+\pi{\bf{0}})+xA\cdot \sigma{\bf{0}}+\pi{\bf{0}}\cdot \sigma{\bf{0}}+g({\bf{0}})+g(\sigma{\bf{0}})\\
 &=[A, \pi{\bf{0}}, \sigma{\bf{0}}A^{T}, \pi{\bf{0}}\cdot \sigma{\bf{0}}+g({\bf{0}})+g(\sigma{\bf{0}})]f(x).
\end{align*}
This shows that $[A, \pi{\bf{0}}, \sigma{\bf{0}}A^{T}, \pi{\bf{0}}\cdot \sigma{\bf{0}}+g({\bf{0}})+g(\sigma{\bf{0}})] \in \operatorname{GB}_m(\F_2)_f$. Thus the correspondence $$(\sigma,\pi)\mapsto [A, \pi{\bf{0}}, \sigma{\bf{0}}A^{T}, \pi{\bf{0}}\cdot \sigma{\bf{0}}+g({\bf{0}})+g(\sigma{\bf{0}})]$$ gives an isomorphism from $\Aut (\mathbb{AD}_f)$ to $\Aut(f)$.

(ii). We first prove that the set $H$ of {affine} permutations $\sigma$ of $\mathbb{F}_{2}^{m}$ such that $f(x)+f(\sigma x)$ is an affine function forms a group. Obviously the identity element is in $H$. If $\sigma_{1}, \sigma_{2} \in H$, then $f(x)+f(\sigma_{1}x)={a_{1} \cdot x+a_{2}}$, $f(x)+f(\sigma_{2}x)={b_{1} \cdot x+b_{2}}$ for some {$a_{1}, b_{1} \in \mathbb{F}_{2}^{m}$ and $a_{2}, b_{2}\in \mathbb{F}_{2}$}. 
We have \begin{eqnarray*}
&f(x)+f((\sigma_{1}\circ \sigma_{2})x)=f(x)+f(\sigma_{2}x)+a_{1}\cdot \sigma_{2}x+a_{2}\\
&=f(x)+f(x)+b_{1} \cdot x+b_{2}+a_{1} \cdot \sigma_{2}x+a_{2}\\
&=a_{1} \cdot \sigma_{2}x+b_{1} \cdot x+a_{2}+b_{2}
\end{eqnarray*} 
is an affine function. Thus $H$ is a group.
{If $\sigma \in \Aut(\mathcal{C}_{\mathbb{AD}_{f}})$, then by the proof of \cite[Theorem 9]{EP}}, $\sigma$ is an affine permutation of $\mathbb{F}_{2}^{m}$ and $f(x)=f(\sigma x)+b' \cdot \sigma x +\varepsilon'$ for some $b' \in \mathbb{F}_{2}^{m}$ and $\varepsilon' \in \mathbb{F}_{2}$. Denote $\sigma x=xA+p$ for some $A \in \mathrm{GL}_{m}(\mathbb{F}_{2})$ and $p \in \mathbb{F}_{2}^{m}$. Then $f(x)=f(xA+p)+b' A^{T} \cdot x+b' \cdot p+\varepsilon'$ and $[A, p, b'A^{T}, b' \cdot p+\varepsilon'] \in \mathrm{GB}_{m}(\mathbb{F}_{2})_{f}$. If $[A, p, b, \varepsilon] \in \mathrm{GB}_{m}(\mathbb{F}_{2})_{f}$, then $f(x)=f(xA+p)+b \cdot x+\varepsilon=f(\sigma x)+b (A^{T})^{-1} \cdot \sigma x+b (A^{T})^{-1} \cdot p+\varepsilon$, where $\sigma x=xA+p$. Since $\sigma$ is an affine function of $\mathbb{F}_{2}^{m}$, then $u \cdot \sigma x$ is an affine Boolean function for any $u \in \mathbb{F}_{2}^{m}$, and thus $u \cdot \sigma x=\pi u \cdot x+\varepsilon_{u}$ for some $\pi u \in \mathbb{F}_{2}^{m}$ and $\varepsilon_{u} \in \mathbb{F}_{2}$. If $\pi u=\pi u'$, then $(u+u') \cdot \sigma x=\varepsilon_{u}+\varepsilon_{u'}$ is a constant function. Since $\sigma $ is a permutation, we obtain $u=u'$ (otherwise, $(u+u') \cdot \sigma x$ is balanced). Thus $\pi$ is a permutation of $\mathbb{F}_{2}^{m}$, and $u \cdot x=\pi^{-1}(u) \cdot \sigma x+\varepsilon_{\pi^{-1}(u)}$ for any $u \in \mathbb{F}_{2}^{m}$. Therefore $\sigma \in \Aut(\mathcal{C}_{\mathbb{AD}_{f}})$. By the above arguments, 
   $ \Aut(\mathcal{C}_{\mathbb{AD}_{f}})
   =\{\sigma: \sigma x=xA+p, x \in \mathbb{F}_{2}^{m}, \text{ where } [A, p, b, \varepsilon] \in \mathrm{GB}_{m}(\mathbb{F}_{2})_{f} \text{ for some } (b, \varepsilon)\in\mathbb{F}^m_2\times\mathbb{F}_2\}
    =\{\sigma \text{ is an affine permutation of } \mathbb{F}_{2}^{m}:  f(x)+f(\sigma x) \text{ is an affine Boolean function of $\F^m_2$}\}.$

(iii). By the proof of Theorem \ref{thmmm}, it is easy to see that 
    $ \Aut(\widetilde{\C}_{D_f})=\{\pi|_{D_{f}}: \pi \text{ is an affine permutation over } \mathbb{F}_{2}^{m} \text{ with } f(x)=f(\pi x)\}.$

(iv). By the proof of Theorem \ref{cor5}, 
  $ \Aut(\mathbb{D}_f)=\{(\pi |_{D_{f}}, (\bar{\pi}^{-1})^{*}|_{{\mathbb{F}_{2}^{m}}^{*}}): \pi \text{ is an affine permutation over } \mathbb{F}_{2}^{m} \text{ with } f(x)=f(\pi x)\},$
     where $\bar{\pi} x=\pi x+\pi {\bf{0}}$, and $(\bar{\pi}^{-1})^{*}$ is the adjoint permutation of $\bar{\pi}^{-1}$.
     \end{proof}


Now we compute the automorphism groups $\Aut(\mathbb{AD}_f), \Aut(\mathbb{TD}_f),$ $\Aut(\mathcal{C}_{\mathbb{AD}_{f}})$ and $\Aut(\widetilde{\C}_{D_f})$ when $f$ is a quadratic plateaued function in $\mathcal{BF}_m$ with nonzero linear structure. Observe that every quadratic Boolean function in $\mathcal{BF}_{m}$ is plateaued \cite{M}, and the function $f$ has no nonzero linear structure if and only if $f$ is bent. Thus every quadratic plateaued function in $\mathcal{BF}_{2m}$ with no nonzero linear structure is equivalent to $f(x,y)=x\cdot y$ for $x,y\in \mathcal{BF}_{m}$ \cite{MS}.

The symplectic group $\operatorname{Sp}_{2m}(\mathbb{F}_2)$ over $\mathbb{F}_2$ consists of 
 all $2m$ by $2m$ matrices $M$ such that $M^TJM=J$ and 
 \[
 J=\begin{pmatrix}
O & I_m & \\
 I_m & O 
\end{pmatrix}.    
 \]
 Here, $I_m$ is the $m\times m$ identity matrix.

\begin{corollary} \label{cor6}
{\rm
Let $f$ be either a quadratic bent function in $\mathcal{BF}_{m}$ with $f^{*} ({\bf{0}})=0$ or an M-M bent function in $\mathcal{BF}_{m}$ with $\deg(g)=3$ and $f^{*} ({\bf{0}})=0$.  Then $\Aut(f), \Aut(\mathbb{TD}_f)$, $\Aut(\mathbb{AD}_f)$, and $\Aut(\mathcal{C}_{\mathbb{AD}_{f}})$ are all isomorphic to $\operatorname{Sp}_m({\mathbb{F}_2})\rtimes\mathbb{F}^m_2$.

}
 \end{corollary}
 \begin{proof}
   By Proposition \ref{thm3.8} and Theorem  \ref{thm10},   $\Aut(\mathbb{TD}_f)$, $\Aut(\mathbb{AD}_f)$ and $\Aut(f)$ are isomorphic. It is known that the design $\mathbb{AD}_f$ of a quadratic bent function $f$ is the symplectic design  with parameters $(\dag)$ \cite[Theorem 2.3]{MPP} and $\Aut(\mathbb{AD}_f)$ is isomorphic to $\operatorname{Sp}_m({\mathbb{F}_2})\rtimes\mathbb{F}^m_2$ \cite{K}. 
 \end{proof}

In Theorem \ref{thm10}, we have shown  $\Aut(\widetilde{\C}_{D_f})$ and $\Aut(\mathbb{D}_f)$. We just wonder the following problem:

\begin{problem}{\rm \label{openpro2}
Determine the automorphism groups $\Aut(\widetilde{\C}_{D_f})$ and $\Aut(\mathbb{D}_f)$ with a simpler form for some special bent functions. 

}
    
\end{problem}



\section{Non-isomorphic designs with old ones}\label{s7}

It is a difficult problem to determine whether or not two designs with identical parameters are {isomorphic}. The following two Theorems \ref{thm4.4} and \ref{thm14} yield new Boolean functions producing two families of a 2-design whose parameters coincide with those of the complement of a point-hyperplane design and a TSDP design, despite being non-isomorphic.

\begin{theorem}\label{thm4.4}{\rm 
	Let $f$ be a bent function in $\mathcal{BF}_m$. For each $b\in \mathbb{F}^m_2$, let $f_b(x)=f(x+b)+f(x)+f(b)$. Then the following statements hold.
	\begin{itemize}
		\item[(i)] A pair $(\mathbb{F}^{m*}_2,\{f_b\in\mathcal{BF}_m:b\in\mathbb{F}^{m*}_2\})$ induces a symmetric 2-design with  parameters  $(2^m-1,(2^{m}-1+(-1)^{f(\mathbf{0})})/2,(2^{m-1}-1+(-1)^{f(\mathbf{0})})/2)$. 
		\item[(ii)] The 2-design in Remark {\ref{rmk3.4}} (with the same parameters as the complement of a point-hyperplane design) and Theorem \ref{thm4.4} {(i)} have the same parameters when $f(\mathbf{0})={0}$. {In this case, they} are {isomorphic with the identity map on point sets and a permutation $\pi$ of $\mathbb{F}^{m*}_2$ on block sets if and only if} $f$ is quadratic. Here, if we extend $\pi$ to $\mathbb{F}^m_2$ with $\pi({\bf{0}})={\bf{0}}$, then $\pi$ is a linear permutation of $\mathbb{F}^m_2$.

        Further, if the translation design derived from $f$ is a non-TSDP design, then they are non-isomorphic. 
	\end{itemize}
	}
\end{theorem}

\begin{proof}
	{{(i).} For any $b\in\mathbb{F}^{m*}_2$, we have that
	\begin{multline*}
		W_{\mathbb{F}^{m*}_2,f_b}(\mathbf{0})=\sum_{p\in \mathbb{F}^{m*}_2}(-1)^{f(b+p)+f(b)+f(p)}\\
		=-(-1)^{f(\mathbf{0})}+(-1)^{f(b)}\sum_{p\in \mathbb{F}^m_2}(-1)^{f(b+p)+f(p)}\\
		=-(-1)^{f(\mathbf{0})}+(-1)^{f(b)}C_f(b)=-(-1)^{f(\mathbf{0})}
	\end{multline*}
	and for any two distinct points $p,q\in\mathbb{F}^{m*}_2$, we have
	\begin{multline*}
		\sum_{b\in \mathbb{F}^{m*}_2}(-1)^{f_b(p)+f_b(q)}\\=\sum_{b\in \mathbb{F}^{m*}_2}(-1)^{f(b+p)+f(b)+f(p)+f(b+q)+f(b)+f(q)}\\
		=-1+(-1)^{f(p)+f(q)}C_f(p+q)=-1.
	\end{multline*} 
	The parameters are evaluated by those three values of summations using  Corollary \ref{cor3.2}.
	
	{(ii).} $(\Longleftarrow)$: Assume that $f$ is quadratic. We can write as $f(p+b)+f(p)+f(b)=p\cdot\pi b+g(b)$ for some function $\pi$ from $\mathbb{F}_2^m$ to itself and Boolean function $g\in\mathcal{BF}_m$. By plugging $p={\bf{0}}$, we obtain $0=f({\bf{0}})=g(b)$. It follows that \begin{align}\label{eqn:(8)}
		f(p+b)+f(p)+f(b)=p\cdot\pi b
	\end{align}
	for all $p,b\in\mathbb{F}^m_2$.
    Due to the proof of Proposition \ref{thm3.8} (i) $\Longrightarrow$ (ii), we can obtain that $\pi$ is a linear permutation of $\F^m_2$, and the result follows. }

	{($\Longrightarrow$): Assume that $f(p+b)+f(p)+f(b)=p\cdot\pi b$ for some permutation $\pi$ of $\mathbb{F}_2^{m*}$ such that if we extend $\pi$ to $\mathbb{F}^m_2$ with $\pi({\bf{0}})={\bf{0}}$, then $\pi$ is a linear permutation of $\mathbb{F}^m_2$.} We define the Boolean function $g$ in $\mathcal{BF}_{2m}$ as
	\begin{align}\label{eqn:(12)}
		{g(x,y)=f(x)+f(y)+f(x+y)+x\cdot\pi y}
	\end{align}
	{for all $x,y\in\mathbb{F}_2^{m*}.$}
	Then $g(p,b)=0$ for all $p,b\in\mathbb{F}^m_2$ by {assumptions}, and so the Boolean $g$ is identically zero, namely $g\equiv 0$. Then $f$ is quadratic using Eq. (\ref{eqn:(12)}) and this proves the first part.

    {For the second part, assume that the $2$-design in Remark \ref{rmk3.4} and Theorem \ref{thm4.4} are isomorphic for some bent function $f$. That is, $f(\mathbf{0})=0$ and $f(p+b)+f(p)+f(b)=\sigma p\cdot\pi b$ for all $p,b\in\mathbb{F}_2^{m*}$ and some permutations $\sigma$ and $\pi$ of $\mathbb{F}_2^{m*}$. By letting $\sigma\mathbf{0}=\pi\mathbf{0}=\mathbf{0}$, we may extend the relation to for any $p,b\in\mathbb{F}_2^m$ and permutation maps $\sigma$ and $\pi$ of $\mathbb{F}_2^m$. Therefore two designs $\mathbb{D}_1=(\mathbb{F}_2^m,{B^{f_b}:b\in \mathbb{F}_2^m})$ associated with Boolean functions $f_b(p)=f(b)+f(p)+f(b+p)$ and $\mathbb{D}_2=(\mathbb{F}_2^m,{B^{g_b}:b\in \mathbb{F}_2^m})$ associated with Boolean functions $g_b(p)=b\cdot p$ are isomorphic. Observe that $f_b(p)=g_{\pi b}(\sigma p)$.} {Assume that the translation design derived from a bent function $f$ is not a TSDP design. Choose three distinct $b_1,b_2,b_3\in\mathbb{F}_2^m$ so that there does not exist $b_4\in \mathbb{F}_2^m$ such that $f_{b_1}(p)+f_{b_2}(p)+f_{b_3}(p)\fallingdotseq f_{b_4}(x)$. However,}
	\begin{eqnarray*}
	&&f_{b_1}(p)+f_{b_2}(p)+f_{b_3}(p)\\
	&=&g_{\pi b_1}(\sigma p)+g_{\pi b_2}(\sigma p)+g_{\pi b_3}(\sigma p)\\
	&=&\sigma p\cdot\pi b_1+\sigma p\cdot\pi b_2+\sigma p\cdot\pi b_3\\
		&=&\sigma p\cdot(\pi b_1+\pi b_2+\pi b_3)=g_{\pi b_1+\pi b_2+\pi b_3}(\sigma p)\\
		&=&f_{{\pi^{-1}}(\pi b_1+\pi b_2+\pi b_3)}(p),
	\end{eqnarray*}
	{which leads to a contradiction.}
\end{proof}

For a Boolean function $f$ in $\mathcal{BF}_m,$ we define its polarization $B_f:\mathbb{F}^m_2\times\mathbb{F}^m_2\rightarrow\mathbb{F}_2$ as $$B_f(x,y)=f(x+y)+f(x)+f(y)+f({\bf{0}}).$$

The following lemma offers an alternative proof of the result \cite[Theorem 3.8]{WFQ}, providing exact statements of the underlying assumptions. 
\begin{lemma} \label{thm6}
{\rm \cite{WFQ}
 Let $f$ be a Boolean function in $\mathcal{BF}_m$. Then the following statements are equivalent.
 \begin{itemize}
     \item [(i)] $f$ is a bent function such that the rank of $(f(x+y))_{x,y\in\mathbb{F}^m_2}$ is $m+2$;
     \item[(ii)] There exist permutations $\sigma, \pi$ of $\mathbb{F}^m_2$ such that $B_f(x,y)=\sigma x\cdot\pi y$ for all $x,y\in\mathbb{F}^m_2$ and both $f\circ\sigma^{-1}$ and $f\circ\pi^{-1}$ are non-affine.
     \end{itemize}
}
\end{lemma}
 \begin{proof}
     
     (i) $\Longrightarrow$ (ii): Since $(f(x+y))_{x,y\in\mathbb{F}^m_2}$ is the point-block incidence matrix of 2-rank $m+2$ with the same parameters $(\dag)$, by {Lemma} \ref{thm3.7}, we have $(f(x+y))_{x,y\in\mathbb{F}^m_2}=(g(\sigma x)+g^*(\pi y)+\sigma x\cdot \pi y)_{x,y\in\mathbb{F}^m_2}$ for some bent function $g$ in $\mathcal{BF}_m$ and some permutations $\sigma, \pi$ of $\mathbb{F}^m_2$. Then $B_f(x,y)=f(x+y)+f(x+{\bf{0}})+f({\bf{0}}+y)+f({\bf{0}})=\bar{\sigma} x\cdot\bar{\pi}y$ for all $x,y\in\mathbb{F}^m_2$, where $\bar{\sigma}x=\sigma x+\sigma{\bf{0}}$ and  $\bar{\pi}x=\pi x+\pi{\bf{0}}$. This proves the first part of (ii). Assume, to the contrary, that $f\circ\bar{\sigma}^{-1}$ is affine. Since $f(x)=g(\sigma x)+g^*(\pi {\bf{0}})+\sigma x\cdot\pi{\bf{0}}$ for all $x\in\mathbb{F}^m_2$, then $f(x)=g(\bar{\sigma} x+\sigma{\bf{0}})+g^*(\pi {\bf{0}})+(\bar{\sigma} x+\sigma{\bf{0}})\cdot\pi{\bf{0}}$, and so $f\circ\bar{\sigma}^{-1}(x)=g(x+\sigma{\bf{0}})+g^*(\pi {\bf{0}})+(x+\sigma{\bf{0}})\cdot\pi{\bf{0}}$. Thus $g$ is affine, which is a contradiction. The same argument shows that  $f\circ\pi^{-1}$ is also non-affine. 
     
     (ii) $\Longrightarrow$ (i): By assumption, the derivative $D_yf(x):=f(x+y)+f(x)$ in direction $y\in\mathbb{F}^{m}_2$ is $D_yf(x)=f(y)+f({\bf{0}})+\sigma x\cdot \pi y$, which is balanced for all $y\in\mathbb{F}^{m*}_2$ by noting $\sigma{\bf{0}}=\pi{\bf{0}}=\bf{0}$. Thus $f$ is bent \cite[Theorem 12]{C1}, \cite{WFQ}. We have
\begin{align*}
  (f_{b}(p))_{b\in \mathbb{F}^m_2}=(f(b+p))_{b\in \mathbb{F}^m_2}\\=(f(b)+f(p)+f({\bf{0}})+\sigma p\cdot \pi b)_{b\in \mathbb{F}^m_2}\\
  =(f(p),\sigma p,1)\begin{pmatrix}
\cdots & 1 & \cdots\\
\cdots & (\pi b)^T & \cdots \\
\cdots & f(b)+f({\bf{0}}) & \cdots
\end{pmatrix}_{b\in\mathbb{F}^m_2}\\=(f(p),\sigma p,1)G.
\end{align*} 
By our assumptions, the dimension of the $\mathbb{F}_2$-linear span of $\{(f(p),\sigma p,1):p\in\mathbb{F}^n_2\}$ is $m+2$ and the rank of $G$ is $m+2$ by using Lemma \ref{lem2}. 
Now the same arguments in {Lemma \ref{lem4.3} (iii) $\Longleftrightarrow$ (iv)} can be applied to complete our proof by using Lemma \ref{lem3.1}.\
 \end{proof}









\begin{corollary} \label{cor7}
    {\rm
   Let $f$ and $g$ be two Boolean functions in $\mathcal{BF}_m$. Then the following statements are equivalent.
   \begin{itemize}
       \item [(i)]  The $2$-design supported by  a point-block incidence matrix $(g(b)+f(p)+b\cdot p)_{b,p\in\mathbb{F}^m_2}$  with parameters  $(2^m,2^{m-1}-2^{\frac{m}{2}-1},2^{m-2}-2^{\frac{m}{2}-1})$ and the $2$-design supported by a point-block incidence matrix $(f(b+p))_{b,p\in\mathbb{F}^m_2}$ with parameters $(2^m,2^{m-1}-(-1)^{{f^{*}}(\mathbf{0})}2^{\frac{m}{2}-1},2^{m-2}{-}(-1)^{{f^{*}}(\mathbf{0})}2^{\frac{m}{2}-1})$ are isomorphic with an identity map on point sets and some linear permutation map on block sets.



        \item[(ii)] There exists a linear permutation $\sigma$ of $\mathbb{F}^m_2$ such that $B_f(x,y)=\sigma x\cdot y$ for all $x,y\in\mathbb{F}^m_2$.

        \item [(iii)] $f$ is a quadratic bent function {with $f^{*}(0)=0$}.
    \end{itemize}
   }
\end{corollary}
\begin{proof}
(i) $\Longleftrightarrow$ (iii): { It follows from Proposition \ref{thm3.8}. }
 (i) $\Longrightarrow$ (ii): It follows from Lemma \ref{thm6}.
     (ii) $\Longrightarrow$ (i): Since $\sigma$ is a linear permutation of $\mathbb{F}^m_2$ and $f$ is bent \cite[Theorem 3.8]{WFQ}, then $f\circ\sigma^{-1}$ is non-affine. By  {Lemmas} \ref{thm3.7} and  \ref{thm6}, the result follows.
\end{proof}

\begin{theorem}\label{thm14}{\rm
Let $g$ and $h$ be bent functions in $\mathcal{BF}_m$ satisfying that
$g+h$ is bent and $(g+h)^*=g^*+h^*$. For each $b\in \mathbb{F}^m_2$, let $f_b(x)=g(b)+h(x+b)+(g+h)(x)$. Then the following statements hold.
\begin{itemize}
    \item [(i)] A pair $(\mathbb{F}^{m}_2,f_b\in\mathcal{BF}_m:b\in\mathbb{F}^{m}_2\})$ induces a symmetric 2-design with parameters $(2^m,2^{m-1}-2^{\frac{m}{2}-1},2^{m-2}-2^{\frac{m}{2}-1})$, denoted this design by $\mathbb{D}_{g,h}$.

    \item[(ii)] The addition design $\mathbb{AD}_f$ derived from a bent function $f\in\mathcal{BF}_m$ and the design $\mathbb{D}_{g,h}$ have the same parameters. If the 2-rank of $\mathbb{TD}_h$ is larger than $m+2$, then $\mathbb{D}_{g,h}$ and $\mathbb{AD}_f$ are not isomorphic.
\end{itemize}

}
\end{theorem}

\begin{proof}
(i). Let $g$ and $h$ be bent functions in $\mathcal{BF}_m$ satisfying that
$g+h$ is bent and $(g+h)^*=g^*+h^*$.
We have 
\begin{multline*}
 C_{g,h}(a)=\frac{1}{2^m}\sum_{x\in\mathbb{F}^m_2}W_g(x)W_h(x)(-1)^{a\cdot x}\\
 =\sum_{x\in\mathbb{F}^m_2}(-1)^{a\cdot x+g^*(x)+h^*(x)}
 =\sum_{x\in\mathbb{F}^m_2}(-1)^{a\cdot x+(g+h)^*(x)}\\
 =2^{\frac{m}{2}}(-1)^{(g+h)(a)}.
\end{multline*}
It follows that for $b\in\mathbb{F}^{m}_2$,  
\begin{multline*}
 W_{\mathbb{F}^m_2,f_b}(\mathbf{0})=\sum_{p\in {\mathbb{F}^{m}_2}}(-1)^{g(b)+h(p+b)+(g+h)(p)}\\
=(-1)^{g(b)}C_{g+h,h}(b)=(-1)^{g(b)}2^{\frac{m}{2}}(-1)^{(g+h)(b)+h(b)}=2^{\frac{m}{2}}
\end{multline*}
and for two distinct points $p,q\in\mathbb{F}^m_2$,
\begin{multline*}
  \sum_{b\in \mathbb{F}^m_2}(-1)^{f_b(p)+f_b(q)}\\=\sum_{b\in \mathbb{F}^m_2}(-1)^{g(b)+h(p+b)+(g+h)(p)+g(b)+h(q+b)+(g+h)(q)}\\
=(-1)^{(g+h)(p)+(g+h)(q)}\sum_{b\in \mathbb{F}^m_2}(-1)^{h(p+b)+h(q+b)}  \\
=(-1)^{(g+h)(p)+(g+h)(q)}{C_{h}}(p+q)=0.
\end{multline*}
The parameters are computed by those three values of summations using  Corollary \ref{cor3.2}.


(ii). Assume, to the contrary, that $\mathbb{D}_{g,h}$ and $\mathbb{AD}_f$ are isomorphic. Then
\begin{align}\label{eqn38}
    f(\sigma p)+f^*(\pi b)+\sigma p\cdot \pi b\nonumber\\=g(b)+h(p+b)+(g+h)(p) \text{ for all } p,b\in\mathbb{F}^m_2
\end{align}
and for some permutations $\sigma, \pi$ of $\mathbb{F}^m_2$. 
By plugging $b={\bf{0}}$ and $p={\bf{0}}$ into Eq. \eqref{eqn38}, respectively,  
\begin{align}\label{39}
    f(\sigma p)+\sigma p\cdot \pi{\bf{0}}=h(p)+(g+h)(p)+f^*(\pi{\bf{0}})+g({\bf{0}}),\\ 
    f^*(\pi b)+\sigma{\bf{0}}\cdot\pi b=g(b)+h(b)+(g+h)({\bf{0}})+f(\sigma{\bf{0}}) \label{40}
\end{align} 
for all $p,b\in\F^m_2$.
By adding Eqs. \eqref{eqn38}, \eqref{39} and \eqref{40}, we have 
\begin{eqnarray}\label{41}
&\bar{\sigma}p\cdot\bar{\pi}b=h(p+b)+h(p)+h(b)+h({\bf{0}})\nonumber\\
&+f(\sigma{\bf{0}})+f^*(\pi{\bf{0}})+\sigma{\bf0}\cdot\pi{\bf0}
\end{eqnarray}
for all $p,b\in\F^m_2$, where $\bar{\sigma}p=\sigma p+\sigma{\bf{0}}$ and $\bar{\pi}b=\pi b+\pi{\bf{0}}$. By plugging $b={\bf{0}}$ into Eq. \eqref{40}, we have 
$\bar{\sigma}p\cdot\bar{\pi}b=h(p+b)+h(p)+h(b)+h({\bf{0}})=B_h(b,p)$ for all $p,b\in\F^m_2$. Since $B_h(b,p)=B_h(p,b)$, then $\tau:=\bar{\pi}\circ\bar{\sigma}^{-1}$ is a linear permutation of $\mathbb{F}^m_2$. From Eq. \eqref{41}, we have 
\begin{eqnarray}\label{42}
& (\tau^*\circ\tau^{-1}b)\cdot p=\tau p\cdot \tau^{-1}b=h(\bar{\pi}^{-1} p+\bar{\sigma}^{-1}b)\nonumber\\
&+h(\bar{\pi}^{-1}p)+h(\bar{\sigma}^{-1}b)+h({\bf{0}}) \end{eqnarray} 
 for all $p,b\in\F^m_2$, where $\tau^*$ is the adjoint permutation of $\tau$.
We claim that both $h\circ \bar{\sigma}^{-1}$ and $h\circ \bar{\pi}^{-1}$ are non-affine. To get a contradiction, we may assume that $h\circ \bar{\pi}^{-1}$ is affine, and so $h\circ \bar{\pi}^{-1}(p)=a\cdot p+\varepsilon$ for some $a\in\mathbb{F}^m_2$ and $\varepsilon\in\mathbb{F}_2$. By Eq. \eqref{42}, we have
\begin{eqnarray*}
    &W_h({\bf{0}})=\sum_{p\in\mathbb{F}^m_2}(-1)^{h(\bar{\pi}^{-1}p+\bar{\sigma}^{-1}b)}\\
    &=(-1)^{h(\bar{\sigma}^{-1}b)+h(\bf{0})}\sum_{p\in\mathbb{F}^m_2}(-1)^{h(\bar{\pi}^{-1}p)+\tau p\cdot\tau^{-1}b}\\
    &=(-1)^{h(\bar{\sigma}^{-1}b)+h(\bf{0})}\sum_{p\in\mathbb{F}^m_2}(-1)^{a\cdot p+\varepsilon+p\cdot (\tau^*\circ\tau^{-1}b)}\\
    &=(-1)^{h(\bar{\sigma}^{-1}b)+h(\bf{0})+\varepsilon}2^m\delta_{a,\tau^*\circ\tau^{-1}b},
\end{eqnarray*}
which contradicts that $h$ is bent. This proves our claim. It then follows from Lemma \ref{thm6} that the 2-rank of $\mathbb{TD}_h$ is $m+2$, which is a contradiction to our assumption.   
\end{proof}


We give an example of the design of Theorem  \ref{thm14} whose 2-rank is not $m+2$.


\begin{example}
{\rm
Let $m=8$ and $\mathbb{F}_{2}^{8}=\mathbb{F}_{2^4} \times \mathbb{F}_{2^4}$. Define $g(x_{1}, x_{2})=\operatorname{Tr}_{1}^{4}(x_{1}x_{2}^{14})$ and 
$ h(x_{1}, x_{2})=\operatorname{Tr}_{1}^{4}(\alpha x_{1}x_{2}^{14})$ for all
$(x_{1}, x_{2}) \in \mathbb{F}_{2^4} \times \mathbb{F}_{2^4}$, where $\operatorname{Tr}^4_1$ is the absolute trace function, and $\alpha$ is the primitive element of $\mathbb{F}_{2^4}$ with $\alpha^{4}+\alpha+1=0$. By \cite[Page 95]{M}, $g, h$ and $g+h$ are bent functions and $(g+h)^{*}=g^{*}+h^{*}$. {By \cite[Theorem 4.3]{WFQ}, the $2$-rank of $\mathbb{TD}_{h}$ is $30$.} Define $f_{(b_{1}, b_{2})}(x_{1}, x_{2})=g(b_{1}, b_{2})+h(x_{1}+b_{1}, x_{2}+b_{2})+(g+h)(x_{1},x_{2})=\operatorname{Tr}_{1}^{4}(b_{1}b_{2}^{14}+\alpha(x_{1}+b_{1})(x_{2}+b_{2})^{14}+(1+\alpha)x_{1}x_{2}^{14})$ for all $(x_{1}, x_{2}) \in \mathbb{F}_{2^4} \times \mathbb{F}_{2^4}$. By Magma, the 2-rank of the design $\mathbb{D}_{g,h}$ in Theorem  \ref{thm14} (i) is $30$, and thus $\mathbb{D}_{g,h}$ is not isomorphic to the addition design $\mathbb{AD}_h$ derived from $h$ by Theorem  \ref{thm14} (ii). 
}
\end{example}

The following open problem   is quite natural.

\begin{problem}{\rm \label{openpro4}

 Prove that the 2-ranks of designs $\mathbb{D}_{g,h}$ and $\mathbb{TD}_h$  in Theorem \ref{thm14}
  are equal.

}
    
\end{problem}


\section{Summary and concluding remarks} \label{s8}

The main results in this paper are the following:

\begin{itemize}
  \item A generic construction for designs from Boolean functions was settled. See Theorem \ref{lem3.1} and Corollary \ref{cor3.2}.


We give an alternative proof for classifying the minimal rank translation designs of a special class of M-M bent functions. See Theorem \ref{thmB3}.

\item  A simple proof of addition designs from $r$-plateaued functions was given, see Lemma \ref{thm 2.23}.

  \item All TSDP  designs with $2^{m-r}$ points were determined in the sense of equivalence, see Theorem \ref{thm4.7}.

  \item Some equivalent relationships between designs, linear codes of designs, and plateaued
functions were discussed. See Theorems \ref{thm4.8}, \ref{thmmm}, \ref{cor5} and Corollary \ref{cor4}. As a byproduct, we settled Open Problems \ref{111} and \ref{222}.

\item Automorphism groups of some addition designs and the linear codes of addition designs were computed. See Theorem \ref{thm10} and Corollary \ref{cor6}.


  \item  More non-isomorphic designs were presented. See Theorems  \ref{thm4.4}, \ref{thm14} and Corollary \ref{cor7}.

  \item Some examples and problems were also presented. See  Open Problems \ref{openproblemnew}, \ref{openpro3}, \ref{openpro1}, \ref{openpro2}, and \ref{openpro4}.

  
\end{itemize}

The interconnections among Boolean functions, designs, and linear codes provide rich opportunities for further research, and  we invite further exploration in this field.

\bigskip
\section*{Acknowledgments}

We would like to thank the associate editor and the anonymous reviewers for their effort in reviewing our manuscript and for providing helpful comments and suggestions.



\newpage

\begin{IEEEbiographynophoto}{Jong Yoon Hyun} received the B.S. degree from Dongguk University in 1997, and the M.S. and Ph.D. degrees in mathematics from Pohang University of Science and Technology (POSTECH) in 2002 and 2006, respectively.
From October 2009 to August 2015, he was a Research Professor at Ewha Womans University, Seoul, South Korea. From September 2015 to March 2019, he worked as a Research Fellow at the Korea Institute for Advanced Study (KIAS), Seoul, South Korea.
He is currently a Professor at Konkuk University (Glocal Campus), Chungju-si, South Korea.
His research interests include coding theory, information theory, cryptographic functions, and algebraic graph theory.

\end{IEEEbiographynophoto}

\begin{IEEEbiographynophoto}{Jieun Kwon} 
 received the Ph.D. degree in mathematics from POSTECH, Pohang, South Korea, in 2022. Since November 2022, she has been with the Semiconductor R\&D Center (CTO), Samsung Electronics. Her current research interests include coding theory, artificial intelligence, and large language models.
\end{IEEEbiographynophoto}

\begin{IEEEbiographynophoto}{Jiaxin Wang}
received the B.S. degree in applied mathematics from Hefei University of Technology,
Hefei, China, in 2017, and the Ph.D. degree in Probability and Mathematical Statistics from Nankai University, Tianjin, China, in 2023. From July 2023 to June 2025, she was a postdoctoral researcher at the Chern Institute of Mathematics, Nankai University, Tianjin, China. Since July 2025, she has been with the School of Mathematics, Hefei University of Technology, Hefei, China. Her current research interests include cryptography and coding theory.
\end{IEEEbiographynophoto}

\begin{IEEEbiographynophoto}{Yansheng Wu} received the Ph.D. degree from Nanjing University of Aeronautics and Astronautics, Nanjing, China, in 2019. From September 2019
to August 2020, he was a Post-Doctoral Researcher with the Department of Mathematics, Ewha Womans University, Seoul, South Korea. Since October
2020, he has been with the School of Computer Science, Nanjing University of Posts and Telecommunications, Nanjing. From March 2023 to February
2024, he was a Visiting Scholar with the Department of Computer Science and Engineering, The Hong Kong University of Science and Technology, Hong Kong. His research interests include coding theory and cryptography.
\end{IEEEbiographynophoto}

\begin{thebibliography}{}


\bibitem{B} T.D. Bending, Bent functions, SDP designs and their automorphism groups. Ph.D. thesis, Queen Mary and Westfield College (1993).

\bibitem{B1} R. E. Block, Transitive groups of collineations of certain designs, Pacific J. Math. 15, 13-19 (1965).

\bibitem{BM} C. Bracken, G. McGuire, Characterization of SDP designs that yield certain spin models, Des. Codes Cryptogr, 36: 45–52 (2005).

\bibitem{BC} L. Budaghyan, C. Carlet, On CCZ-equivalence and its use in secondary constructions of bent functions, in Preproceedings of the International Workshop on Coding and Cryptography, WCC 2009, Ullensvang, Norway,  pp. 19–36 (2009). 

\bibitem{BC1}  L. Budaghyan, C. Carlet, CCZ-equivalence of bent vectorial functions and related constructions, Des. Codes Cryptogr. 59, 69–87 (2011).


\bibitem{CS} P. J. Cameron and J. J. Seidel, Quadratic forms over GF(2), Indag. Math. 35, 1-8 (1973).



\bibitem{C1}  C. Carlet, “Boolean functions for cryptography and error correcting codes,” in Boolean Models and Methods in Mathematics, Computer Science, and Engineering, P. L. Hammer and Y. Crama, Eds. Cambridge, U.K.: Cambridge Univ. Press, 2010.

\bibitem{C2} C. Carlet, “Vectorial Boolean functions for cryptography,” in Boolean Models and Methods in Mathematics, Computer Science, and Engineering, P. L. Hammer and Y. Crama, Eds. Cambridge, U.K.: Cambridge Univ. Press (2010).


\bibitem{C3} C. Carlet. Boolean Functions for Cryptography and Coding Theory. Cambridge Univ. Press, 562 pages, 2021 




\bibitem{CM} C. Carlet and S. Mesnager, “Four decades of research on bent functions,” Des., Codes Cryptogr. 78, no. 1: 5–50 (2016).


\bibitem{CD} C. J. Colbourn and J. H. Dinitz, CRC Handbook of Combinatorial Designs, CRC Press, Boca Raton, FL (2007).

\bibitem{DN} U. Dempwolff, T. Neumann, Geometric and design-theoretic aspects
of semibent functions I, Des. Codes Cryptogr., 57, 373–381 (2010).

\bibitem{DC}  C. Ding, A construction of binary linear codes from Boolean functions, Disc. Math., 339(9): 2288-2303 (2016).




\bibitem{DMT} C. Ding,  et al., Bent vectorial functions, codes and designs,   IEEE Trans. Inf. Theory, 65, no. 11:  7533-7541 (2019).



\bibitem{DT} C. Ding, and C. Tang, Combinatorial $t$-designs from special functions, Cryptogr. Commun., 12, no. 5: 1011-1033 (2020).

\bibitem{DT0} C. Ding, and C. Tang, Designs from Linear codes, Singapore: World Scientific (2022).

\bibitem{DT1} C. Ding, C. Tang,  Infinite families of near MDS codes holding $t$-designs, IEEE Trans. Inf. Theory, 66(9): 5419-5428  (2020).

\bibitem{D} J.F. Dillon. Elementary Hadamard difference sets. University of Maryland, College Park, 1974.





\bibitem{DS} J. F. Dillon and J. R. Schatz, Block designs with the symmetric difference property, in Proc. NSA Mathematical Sciences Meetings'' (R. L. Ward, Ed.), pp. 159-164, U.S. Govt. Printing Office, Washington, DC (1987).



\bibitem{EP} Y. Edel, A. Pott, On the equivalence of nonlinear functions, Enhancing Cryptographic Primitives with Techniques from Error Correcting Codes,  pp. 87103 (2009).


\bibitem{HP} S. Hodžić, E. Pasalic, Y. Wei, F. Zhang, Designing Plateaued Boolean Functions in Spectral Domain and Their Classification, IEEE Trans. Inf. Theory, 65, no. 9: 5865-5879 (2019).



\bibitem{K} W. M. Kantor, Symplectic Groups, Symmetric Designs, and Line Ovals, J. Algebra 33, 43–58 (1975).

\bibitem{K1} W. M. Kantor, Exponential number of two-weight codes, difference sets and
symmetric designs, Disc. Math. 46, 95–98 (1983).

\bibitem{K2} W. M. Kantor, Classification of 2-transitive symmetric designs, Graphs and Combinatorics 1 (1), 165–166 (1985).

\bibitem{L}  E.S. Lander, Symmetric Designs: An Algebraic Approach, London Mathematical Society Lecture Note Series, Cambridge University Press (1983).

\bibitem{MW} G. McGuire, H. N. Ward,
Characterization of Certain Minimal Rank Designs. J. Comb. Theory, Ser. A 83(1): 42-56 (1998).



\bibitem{M1} R. L. McFarland. A family of noncyclic difference sets. J. Combinatorial Theory, Ser. A, vol. 15, pp. 1–10 (1973).

\bibitem{MS1} F. J. MacWilliams and N.J.A. Sloane, The Theory of Error-Correcting Codes, North-Holland Publishing Company, 1977.

\bibitem{Ma} J. L. Massey, Linear codes with complementary
duals, Discrete Math. 106/107, 337-342 (1992).



\bibitem{MPP} W. Meidl, A. Polujan, and A. Pott, Linear codes and incidence structures of bent functions and their generalizations. Discrete Math., 346(1), 113157,(2023). 


\bibitem{MZYC} Q. Meng, H. Zhang, et al., On the degree of homogeneous bent functions. Discret. Appl. Math. 155, 665–669 (2007).

\bibitem{M} S. Mesnager, Bent Functions. Fundamentals and Results, Springer, 2016.



\bibitem{MS} S. Mesnager, and A. Sınak, Several classes of minimal linear codes with few weights from weakly regular plateaued functions, IEEE Trans. Inf. Theory, 66, no. 4: 2296-2310 (2019).











 





\bibitem{P} A.  Polujan, Boolean and vectorial functions: A design-theoretic point of view, PhD dissertation,  Otto-von-Guericke-University at Magdeburg  (2021).


\bibitem{PP}  A.  Polujan and A. Pott, On design-theoretic aspects of Boolean and vectorial bent function,   IEEE Trans. Inf. Theory, 67, no. 2: 1027-1037 (2020).


\bibitem{R} O.S. Rothaus, On bent functions. J. Combin. Theory Ser. A 20(3): 300–305 (1976).




\bibitem{T} C. Tang, Infinite families of 3‐designs from APN functions, J. Combin. Des., 28, no. 2: 97-117 (2020).


\bibitem{TD} C. Tang, C. Ding,  An infinite family of linear codes supporting 4-designs, IEEE Trans. Inf. Theory, 67(1): 244-254 (2020). 



\bibitem{WH} P. Wang, Z. Heng, Self-orthogonal codes from plateaued functions, Chinese Journal of Electronics, 34(5): 1483-1496 (2025).

\bibitem{WTD} X. Wang, C. Tang,  C. Ding,  Infinite families of cyclic and negacyclic codes supporting 3-designs, IEEE Trans. Inf. Theory, 69, no. 4: 2341-2354 (2022). 



\bibitem{WFQ} G. Weng, R. Feng, W. Qiu, On the ranks of bent functions,  Finite Fields Appl., 13, 1096–1116 (2007). 

\bibitem{WFQZ}G. Weng, R. Feng, W. Qiu,  Z. Zheng,   The ranks of Maiorana-McFarland bent functions, Sci. China Ser. A Math.‌, 51(9), 1726-1731 (2008).




\bibitem{XLW}  C. Xiang, X. Ling, et al., Combinatorial $t$-designs from quadratic functions,  Des. Codes Cryptogr. 88: 553-565 (2020).

\end{thebibliography}
\end{document}